\documentclass[a4,article,nofootinbib, twocolumn]{revtex4}

\usepackage{graphicx}
\usepackage{dcolumn}
\usepackage{bm}
\usepackage{color}
\usepackage{amsmath}%
\usepackage{amsmath}%
\usepackage{amsfonts}%
\usepackage{amssymb}%
\usepackage{breqn}
\usepackage{fixfoot}

\usepackage{bm}
\usepackage{mathrsfs}
\usepackage{amsthm}
\usepackage{breqn}
\usepackage{doi}
\usepackage{float}
\usepackage[normalem]{ulem}
\usepackage{breqn}

\newcommand{\mc}{\mathcal}
\newcommand{\cp}{\times}

\newcommand{\cu}{\nabla\times} 
\newcommand{\JI}[1]{\bol{#1}\cdot\cu{\bol{#1}}}
\newcommand{\BN}[1]{\bol{#1}\cp\left(\cu{\bol{#1}}\right)}

\newcommand{\bol}{\boldsymbol}
\newcommand{\hb}[1]{\hat{\bol{#1}}}
\newcommand{\abs}[1]{\left\lvert{#1}\right\rvert}
\newcommand{\w}{\wedge}
\newcommand{\lr}[1]{\left({#1}\right)}

\newcommand{\av}[1]{\left\lvert{#1}\right\rvert}
\newcommand{\mf}{\mathfrak}
\newcommand{\p}{\partial}
\newcommand{\mJ}{\mathcal{J}}

\newtheorem{theorem}{\textit{Theorem}}[section]

\makeatletter
\let\cat@comma@active\@empty
\makeatother

\begin{document}
\title{Diffusion with finite-helicity field tensor: a new mechanism of generating heterogeneity}
\author{N. Sato and Z. Yoshida}
\affiliation{Graduate School of Frontier Sciences, The University of Tokyo,
Kashiwa, Chiba 277-8561, Japan}
\date{\today}

\begin{abstract}
Topological constraints on a dynamical system often manifest themselves as breaking of the Hamiltonian structure;
well-known examples are non-holonomic constraints on Lagrangian mechanics.
The statistical mechanics under such topological constraints is the subject of the present study.
Conventional arguments based on \emph{phase spaces}, \emph{Jacobi identity}, \emph{invariant measure}, or the \emph{H theorem} are no longer applicable, since all these notions stem from the symplectic geometry underlying canonical Hamiltonian systems.
Remembering that Hamiltonian systems are endowed with field tensors (canonical 2-forms) that have zero helicity,
our mission is to extend the scope toward the class of systems governed by finite-helicity field tensors.
Here we introduce a new class of field tensors that are characterized by Beltrami vectors.
We prove an H theorem for this Beltrami class. 
The most general class of energy-conserving systems are non-Beltrami,
for which we identify the ``field charge'' that prevents the entropy to maximize, resulting in creation of heterogeneous distributions.
The essence of the theory can be delineated by classifying three-dimensional dynamics.
We then generalize to arbitrary (finite) dimensions.
\end{abstract}


\maketitle

\begin{normalsize}

\section{Introduction}
While Fick's law is amenable to the intuition telling that diffusion will gradually remove gradients in distributions,
we do find many counter-examples where diffusion generates or sustains gradients
(so-called ``up-hill'' or ``inward'' diffusion).
Indeed, the theoretical guarantee for the minimization of gradients (or maximization of entropy) is rather limited;
conventional arguments start from the identification of a phase space and an invariant measure (Liouville's theorem),
by which one may construct an H theorem to give presumption of the ergodiv hypothesis.
Usually, these deductions rely on the Hamiltonian structure of underlying microscopic dynamics.
Given a general (non-Hamiltonian) system, therefore,
one should once abandon the hypothesis of maximum entropy,
and study different conditions by which diffusion may diminish or generate gradients.

In this work, we propose a new paradigm of dynamics by which the regime of the maximum-entropy law is extended beyond Hamiltonian systems;
the new regime is identified by the ``Beltrami condition'' that demands vanishing of what we will call ``field charge''.
We will show that the field charge is the root cause of inhomogeneity that can persist against diffusion.

Before nailing down the target of analysis, we make a short review of the theories of up-hill diffusion.
There are two different causes of such phenomena; one is the energy and the other is the geometry of space.
If the energy of a system includes some term that works to attract particles, 
the ``Boltzmann distribution'' explains the heterogeneity in the thermal equilibrium.
Gravitational systems are such examples.
Chemical potentials also play a similar role in grand-canonical ensembles.
However, our interest is in the second kind of systems where the energy is just simple (for example, a convex function)
but the space is skewed. 
So-called ``topological constraints'' limit the effective space of dynamics, resulting in heterogeneous distributions in the \emph{a priori} space.
For Hamiltonian systems, the Casimir invariants (which originate from the center of the Poisson algebra)
foliate the phase space (such Hamiltonian systems are said noncanonical)\,\cite{Morrison1982,Littlejohn1982,Morrison1998}.
The Boltzmann distribution on a Casimir leaf may be viewed as a grand-canonical distribution with a chemical potential multiplying on the Casimir invariant, i.e. a Casimir invariant may be regarded as an action variable\,\cite{Yoshida2014,Yoshida2016}.
In the self-organization of a magnetospheric plasma,
the magnetic moment of charged particles plays the role of the Casimir invariant\,\cite{Yoshida2010}.
As far as the system is Hamiltonian, the effective phase space is a (locally) symplectic leaf, 
so that the standard methods of statistical mechanics are readily applicable.
We can formulate a Fokker-Planck equation to simulate the diffusion in magnetospheric systems\,\cite{Sato1,Sato3,Sato4}.
The key reserved for Hamiltonian systems is the ``integrability'' of the topological constraints,
which, however, is no longer valid for non-Hamiltonian systems.
This is the regime of our interest.

Here we assume the constancy of energy in the autonomous limit,
i.e., motion occurs in the direction perpendicular to the gradient of the energy.
The statistical dynamics is driven by a white noise in the energy.
When the topological constraints are \textit{non-integrable} (in the sense of Frobenius' theorem \cite{Frankel,Agricola,Kobayashi}), 
there is no way to construct symplectic leaves on which we can define a canonical phase space. 
Mathematically, the non-integrability is equivalent to the failure 
of the Jacobi identity \cite{Bloch_2,Shaft_1,Bates_1}, with critical consequences for the dynamics \cite{Chandre}. 
Non-integrable constraints occur, for example, in nonholonomic mechanical systems \cite{Bloch_1,deLeon_1}, such as the rolling of a disk without slipping on a horizontal surface.
In addition to nonholonomic mechanics \cite{deLeon_1,Chaplygin} and molecular dynamics \cite{Sergi_1,Sergi_2,Tuckerman,Erza}, it will be shown that other systems, such as the $\bol{E}\cp\bol{B}$ drift equation of plasma dynamics \cite{Cary,Littlejohn_2} and the Landau-Lifshitz equation \cite{Landau_2,Gilbert} for the magnetization of a ferromagnetic material, fall in this category. 

The essence of the theory can be delineated by three-dimensional mechanics.
The velocity $\bm{v}$ of motion can be written as
\begin{equation}
\bm{v} = \bm{w}\times\nabla H,\label{EoM3D}
\end{equation}
where $H$ is the energy, and $\bm{w}$ is a fixed vector such that the velocity is perpendicular to $\nabla H$.
The operation of $\bm{w}\times$ can be represented by an antisymmetric operator $\mc{J}$ that we call a field tensor.
For (\ref{EoM3D}) to be Hamiltonian, $\bm{w}$ must be ``helicity free'' ($\bm{w}\cdot\nabla\times\bm{w}=0$),
and then, $\bm{w}$ is integrable;
we may write $\bm{w}=\lambda\nabla C$ with some $\lambda$ and $C$, and $C =$ const gives the Casimir leaves. 
The three-dimensional Lie-Poisson algebras are classified by the Bianchi classification;
for the complete list of symplectic leaves, see \cite{Yoshida2017}.
However, the target of our study are systems where $\bm{w}$ has finite-helicity. 
We define the ``Beltrami class'' by those $\bm{w}$ 
such that $\nabla\times\bm{w}=\gamma\bm{w}$ with $\gamma\neq0$.
In Sec. III and IX we will prove an H theorem for the Beltrami class.
We will also show that the ``field charge'' that is measured by $\nabla\cdot\left[\bm{w}\times\lr{\nabla\times\bm{w}}\right]$
(hence, the Beltrami class is charge-free) causes heterogeneity.
Notice that the mechanism of creation of heterogeneity is totally different from the aforementioned ones operated by some attracting potential energy, or the foliation of the phase space.
In Sec. IV we will give numerical demonstration of the effect of the field charge.
We will then generalize the theory to arbitrary $(>2)$ dimensions in Sec. V-IX.

\section{Conservative dynamics in three dimensions}

The simpler and instructive $3$-dimensional ($3$D) case is first discussed.
In its general form, the equation of motion of a $3$D conservative system is given by \eqref{EoM3D}. 
Here, $\bol{v}=\dot{\bol{x}}$ is the velocity in the Cartesian coordinate system $\bol{x}=\lr{x,y,z}$ of $\mathbb{R}^{3}$, the vector field $\bol{w}=\bol{w}\lr{\bol{x}}$ (assumed smooth and non-vanishing) serves as antisymmetric operator (to be defined later), and the real valued smooth function $H$ represents the Hamiltonian function. Evidently, $\dot{H}=0$. 
However, system (\ref{EoM3D}) is not, in general, Hamiltonian. As already mentioned, the condition is given by the Jacobi identity, which demands that $\bol{w}$ has vanishing helicity density:
\begin{equation}
h=\bol{w}\cdot\lr{\nabla\cp\bol{w}}=0.\label{JI3D}
\end{equation}
The validity of (\ref{JI3D}), which determines whether $\bol{w}$ qualifies as a Poisson operator, is related to the existence of additional integral invariants and 
to the availability of an invariant measure. Indeed, the following conditions are locally equivalent: for some open set $U\subset\mathbb{R}^{3}$,
\begin{subequations}\label{eq3}
\begin{align}
&1.~\bol{w}\cdot\lr{\nabla\cp\bol{w}}=0~in~U,\\
&2.~\exists~\lambda,C:U\rightarrow\mathbb{R}:~\bol{w}=\lambda\nabla C~in~U,\\
&3.~\exists~g\neq 0,~g:U\rightarrow\mathbb{R}:~\nabla\cdot\lr{g\bol{v}}=0~\forall H~in~U.
\end{align}
\end{subequations}
$\lr{1\implies 2}$ is the Frobenius integrability condition 
for the vector field $\bol{w}$ (see \cite{Frankel,Agricola,Kobayashi}). Then locally we can find two functions $\lambda$ and $C$ such that $\bol{w}=\lambda\nabla C$. 
$\lr{2\implies 1}$ is trivial since $\bol{w}\cdot\lr{\nabla\cp\bol{w}}=-\lr{\lambda\nabla C\cdot\nabla\lambda\times\nabla C}=0$. $\lr{2\implies 3}$ can be verified by observing that:
\begin{equation}
\nabla\cdot\lr{g\bol{v}}=0~\forall H~
\iff~\nabla H\cdot\nabla\times\lr{g\bol{w}}=0
~\forall H.
\end{equation}
The implication follows by setting $g=\lambda^{-1}$.
$\lr{3\implies 2}$ If there is an invariant measure $g$ for any $H$, then $\nabla\times\lr{g\bol{w}}=\bol{0}$. Therefore $\bol{w}=g^{-1}\nabla C$ on $U$.

The function $C$, called a Casimir invariant, 
is a constant of motion for any choice of $H$ and poses an integrable
topological constraint on the dynamics. If $\bol{w}$ cannot be expressed in terms of a Casimir invariant, the dynamics is still constrained by the condition $\bol{w}\cdot \bol{v}=0$, which then represents
a non-integrable topological constraint.  

To introduce a classification of conservative dynamics beyond Hamiltonian systems, 
we define the \textit{field force}:
\begin{equation}
 \bol{b}=\bol{w}\cp\lr{\nabla\cp\bol{w}},
\end{equation}
and the \textit{field charge}:
\begin{equation}
\mf{B}=4\nabla\cdot\bol{b}=4\nabla\cdot\left[\bol{w}\cp\lr{\nabla\cp\bol{w}}\right],\label{Bness3D}
\end{equation}
This naming was chosen by analogy with electromagnetism: when $\bol{w}$ is the antisymmetric operator associated to the $\bol{E}\cp\bol{B}$ drift motion \cite{Cary} of a charged particle in a magnetic field $\bol{B}$ of constant strength, the vector $\bol{b}$ is the magnetic force $\bol{B}\cp\lr{\cu\bol{B}}$. 
In fact, the drifting velocity is given by $\bol{v}=\bol{E}\cp\bol{B}/{B^{2}}$, with $\bol{E}=-\nabla\phi$ the electric field and $\phi$ the electrostatic potential. Hence, the antisymmetric operator is $\bol{w}=\bol{B}/B^{2}$, the Hamiltonian $H=\phi$, and the Jacobi identity holds when $\bol{B}\cdot\nabla\cp\bol{B}=0$.
To understand the geometrical meaning of $\bol{b}$ the following vector identity $\bol{b}=\BN{w}=\nabla w^{2}/2-\lr{\bol{w}\cdot\nabla}\bol{w}$ is useful. Using this formula for $\hb{w}=\bol{w}/w$, we have $\hb{b}=-\lr{\hb{w}\cdot\nabla}\hb{w}=-\hb{k}$, where $\hb{k}$ is the curvature vector. Therefore, $\bol{b}$ is related to the curvature of $\bol{w}$. Furthermore, observe that the curl of a vector field $\bol{w}$ admits the decomposition:
\begin{equation}
\nabla\times\bol{w}=\frac{\bol{b}\times\bol{w}+h\bol{w}}{w^{2}}.\label{BJDec}
\end{equation}

Three dimensional conservative systems are then classified according to figure \ref{fig1}.
In the next section, the statistical relevance of this classification will be made clear. 

\begin{figure}[H]
\hspace*{-0cm}\centering
\includegraphics[scale=0.195]{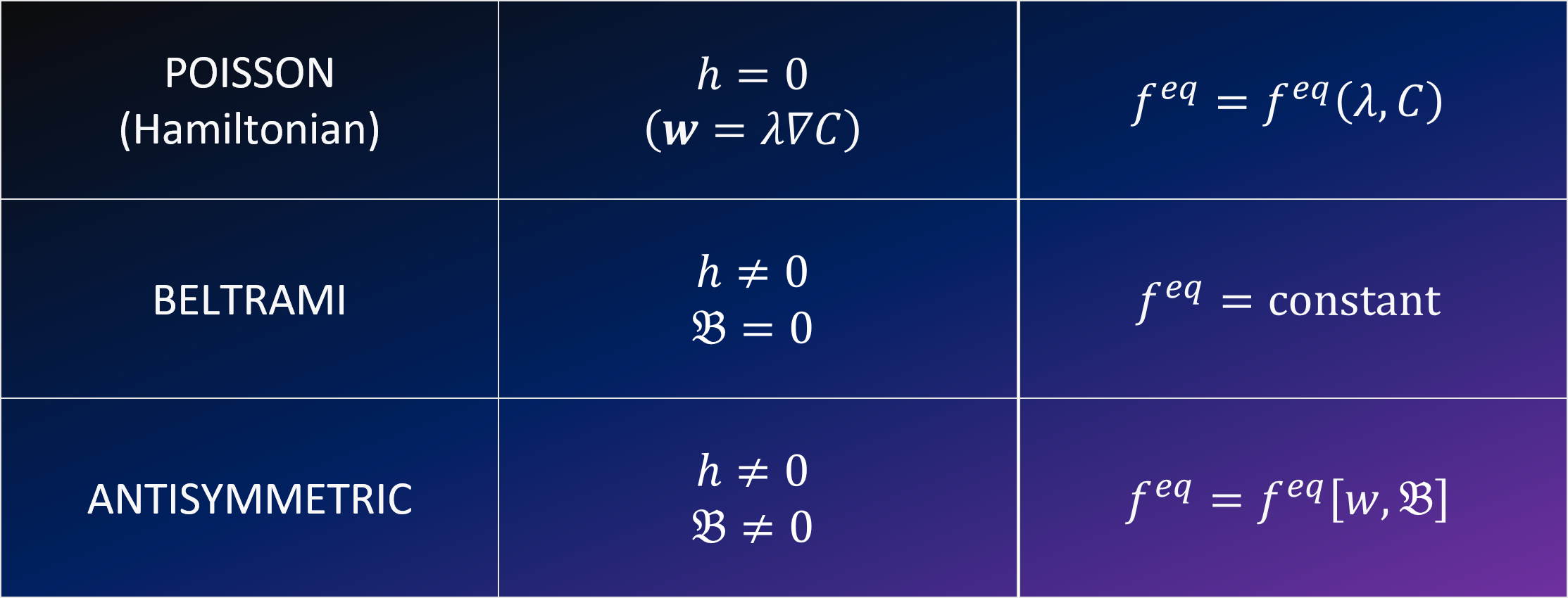}
\caption{\footnotesize Classification of 3D conservative dynamics. The right column shows the equilibrium distribution function $f^{eq}$ of an ensemble of particles obeying \eqref{EoM3D} when $\nabla H$ is a white noise process (see section III). The square bracket in the last column indicates that the dependence of $f^{eq}$ on $w$ and $\mf{B}$ is not necessarily a functional one.}
\label{fig1}
\end{figure}	

\section{Diffusion in three dimensions}

To examine the properties of diffusion, consider the purely stochastic equation of motion with $\nabla H=\bol{\Gamma}$:
\begin{equation}
\bol{v}=\bol{w}\times\bol{\Gamma},\label{XRan2}
\end{equation}
where $\bol{\Gamma}=\lr{\Gamma_{x},\Gamma_{y},\Gamma_{z}}$ is $3$-dimensional white noise.
If this were a conventional diffusion process, one would expect the density distribution $f$ of an ensemble
of particles obeying (\ref{XRan2}) to become progressively flat. 
This is not necessarily the case. 
To see this, consider the Fokker-Planck equation (to be derived later) associated to the stochastic differential equation (\ref{XRan2}):
\begin{dmath}
\frac{\p f}{\p t}=\frac{1}{2}\nabla\cdot\left[\bol{w}\times\lr{\nabla\times f\bol{w}}\right]=\frac{1}{2}\lr{\Delta_{\perp}f+\nabla f\cdot\bol{b}+\frac{1}{4}f\mf{B}}\label{FPEDiff3D_3}.
\end{dmath}	
Here, we introduced the \textit{normal Laplacian} $\Delta_{\perp}f=\nabla\cdot\left[\bol{w}\times\lr{\nabla f\times\bol{w}}\right]$. The word \textit{normal} refers to the fact that its value only depends on the component of $\nabla f$ perpendicular to $\bol{w}$, $\nabla_{\perp}f=\bol{w}\cp\lr{\nabla f\cp\bol{w}}/w^{2}$. 
In the following we shall always assume $f$ to be a classical solution to the diffusion equation that
admits all necessary derivatives. 

The stationary form of equation \eqref{FPEDiff3D_3} is a non-elliptic PDE (see \cite{Gilbarg,Evans,Brezis} for the definition of ellipticity). Hence, the existence of a unique solution is not trivial. As it will be shown in the following, the nature of the stationary solution changes depending on the geometric properties of $\bol{w}$. 

For $f$ to become flat, the diffusion process (\ref{XRan2}) must maximize Shannon's information entropy:
\begin{equation}
S=-\int_{\Omega}{f\log{f}}\,dV.
\end{equation}
Here, $\Omega\subset\mathbb{R}^{3}$ is a smoothly bounded domain occupied by the statistical ensemble, and $dV=dxdydz$ is the volume element in $\mathbb{R}^{3}$. However, for a given $\bol{w}$, $S$ is not necessarily maximized. When $h=0$ the system is Hamiltonian and from \eqref{eq3} it follows that the
invariant measure is $\lambda^{-1} dV$ ($\lambda \neq 0$ since $\bol{w}\neq \bol{0}$). Then, as one may expect, the appropriate entropy is:
\begin{equation}
\Sigma_{\lambda}=-\int_{\Omega}{f\log{\lr{f\lambda}}}\,dV,\label{S3D}
\end{equation} 
which is equivalent to $S$ only if $\lambda=$ constant. 
In fact, using \eqref{FPEDiff3D_3} and assuming the boundary condition $\bol{w}\cp\lr{\nabla\cp f\bol{w}}\cdot\bol{n}=0$ on $\p\Omega$, with $\bol{n}$ the unit outward normal to $\p\Omega$, it follows that:
\begin{equation}
\frac{d\Sigma_{\lambda}}{dt}=\frac{1}{2}\int_{\Omega}{f\lambda\abs{\lambda\nabla C\cp\nabla\log\lr{f\lambda}}^2}\,\lambda^{-1}dV\geq 0.\label{dSdt_0}
\end{equation} 
Assuming that $f>0$ in $\Omega$ and observing that $dS/dt$ must vanish in the limit $t\rightarrow\infty$, one sees that:
\begin{equation}
f^{eq}=\lim_{t\rightarrow\infty}f=\frac{A}{\lambda}\exp{\left\{-\gamma\mc{F}\lr{C}\right\}}~in~\Omega,\label{feq3Dh0}
\end{equation}
where $\mc{F}\lr{C}$ is an arbitrary function of the Casimir invariant $C$ determined by the initial conditions, $A>0$ and $\gamma>0$ real constants. 

It is a pivotal point of the present study the proof that the maximization of $S$ for $h\neq 0$
depends on the behavior of the field charge $\mf{B}$. Indeed, the following result holds:

\begin{theorem}\label{thmIII.1}
Let $\bol{w}$ be a smooth vector field on a smoothly bounded domain $\Omega\subset\mathbb{R}^{3}$ with boundary $\p\Omega$. Consider equation \eqref{FPEDiff3D_3} for $f>0$ in $\Omega$ with boundary conditions $\bol{b}\cdot\bol{n}=0$ and $\bol{w}\cp\lr{\nabla f\cp \bol{w}}\cdot\bol{n}=0$ on $\p\Omega$. Assume $\mf{B}=0$ and $h\neq 0$ in $\Omega$. Then,
\begin{equation}
\lim_{t\rightarrow\infty}\nabla f=\bol{0}~in~\Omega.
\end{equation}
\end{theorem}

\begin{proof}
Using equation (\ref{FPEDiff3D_3}) and the boundary conditions, the rate of change in the entropy \eqref{S3D} reads:
\begin{equation}
\frac{dS}{dt}=\frac{1}{2}\int_{\Omega}{f\left[-\frac{\mf{B}}{4}+\av{\bol{w}\cp\nabla \log{f}}^{2}\right]}\,dV.\label{dSdt_2_1}
\end{equation}
Since by hypothesis $\mf{B}=0$, we must have $\lim_{t\rightarrow\infty}\bol{w}\cp\nabla f=\bol{0}$ in $\Omega$. 
Furthermore, $h\neq 0$ implies that $\bol{w}$ is not integrable, i.e. there is no Casimir invariant $C$ such that $\bol{w}=\lambda \nabla C$ for some function $\lambda$. 
Hence, if we could satisfy $\nabla f=\alpha\bol{w}$ in $\Omega$ for some function $\alpha\neq 0$, this would contradict the non-integrability of $\bol{w}$.  
Therefore, $\nabla f=\bol{0}$ in $\Omega$ when $t\rightarrow\infty$. 
\end{proof}

The boundary conditions used to derive equations \eqref{dSdt_0} and \eqref{dSdt_2_1} ensure
the thermodynamical closure of the system by avoiding loss of probability through the boundaries and will
be discussed in more detail later. It is also worth noticing that, if $\mf{B}\neq 0$, $f=$ constant is not a stationary solution of \eqref{FPEDiff3D_3}, as one can verify by substitution.
Indeed, one obtains the condition $\mf{B}=0$. 
An operator $\bol{w}$ satisfying such property will be called a \textit{Beltrami operator} (remember fig. \ref{fig1}).
This name refers to the Beltrami condition $\bol{b}=\BN{w}=\bol{0}$, which describes vectors aligned with their own vorticity, resulting in $\mf{B}=0$.

The determination of the stationary solution to \eqref{FPEDiff3D_3} in the remaining case where $h\neq 0$ and $\mf{B}$ is allowed to take non-zero values in $\Omega$ requires the machinery of functional analysis and will not be discussed here as this mathematical issue goes beyond the scope of the present paper. 
However, the special case in which the field force $\hb{b}=\hb{w}\cp\lr{\nabla\cp\hb{w}}$ 
of the normalized vector $\hb{w}=\bol{w}/w$ can be expressed by means of a scalar potential as $\hb{b}=\nabla\zeta$ can be solved explicitly and provides a concrete example of how self-organization in 
non-Hamiltonian system is intrinsically different from the foliation by Casimir invariants obtained in 
\eqref{feq3Dh0}. To see this, consider the entropy:
\begin{equation}
\Sigma_{\zeta}=-\int_{\Omega}{f\left[\log\lr{fw}+\zeta\right]}\,dV,
\end{equation} 
and assume the boundary condition $\bol{w}\cp\lr{\nabla\cp f\bol{w}}\cdot\bol{n}=0$ on $\p\Omega$.
Then, the rate of change in $\Sigma_{\zeta}$ takes the form:
\begin{equation}
\frac{d\Sigma_{\zeta}}{dt}=\frac{1}{2}\int_{\Omega}{f\abs{\bol{w}\cp\nabla\left[\zeta+\log\lr{fw}\right]}^{2}}\,dV\geq 0.
\end{equation}
Since by hypothesis $h\neq 0$, it follows that:
\begin{equation}
f^{eq}=\lim_{t\rightarrow\infty}f=\frac{A}{w}e^{-\zeta}~in~\Omega.\label{feqbdz}
\end{equation}
Here, $A>0$ is a real constant. Notice how $f^{eq}$ is determined by the field charge $\hat{\mf{B}}=\Delta\zeta$ and the strength $w=\abs{\bol{w}}$. 

\section{Numerical simulation}

It is now useful to make qualitative considerations on how the orbit of a conservative particle obeying (\ref{EoM3D}) is modified by the introduction of random noise. First, consider the Euler rotation equation for a rigid body. In this case $\bol{w}=\bol{x}$, with $\bol{x}$ the angular momentum, and the Hamiltonian is $H_{0}=\lr{x^{2}I_{x}^{-1}+y^{2}I_{y}^{-1}+z^{2}I_{z}^{-1}}/2$ with $I_{x}$, $I_{y}$, and $I_{z}$ the momenta of inertia. $\bol{w}$ is a Poisson operator because the Jacobi identity is satisfied: 
$\bol{x}\cdot\nabla\times\bol{x}=0$. As a consequence, the total angular momentum $C=\bol{x}^{2}/2$ is a Casimir invariant. The unperturbed orbit of the rigid body, given by the intersection of the integral surfaces 
$H_{0}$ and $C$, is given in figure \ref{fig2}(a).
Now, we perturb the Hamiltonian $H_{0}$ so that the force acting on the particle becomes $\nabla H=\nabla H_{0}+\bol{\Gamma}$. The resulting stochastic differential equation is:
\begin{equation}
\bol{v}=\bol{x}\times\lr{\nabla H_{0}+\bol{\Gamma}}.\label{X3}
\end{equation}
Clearly, the energy $H_{0}$ is not anymore a constant of motion. However, the Casimir invariant $C$ is unaffected by the perturbations.
The result is a random process on the level set $C=$constant (see figure \ref{fig2}(b)). 

\begin{figure}[h]
\hspace*{-0cm}\centering
\includegraphics[scale=0.28]{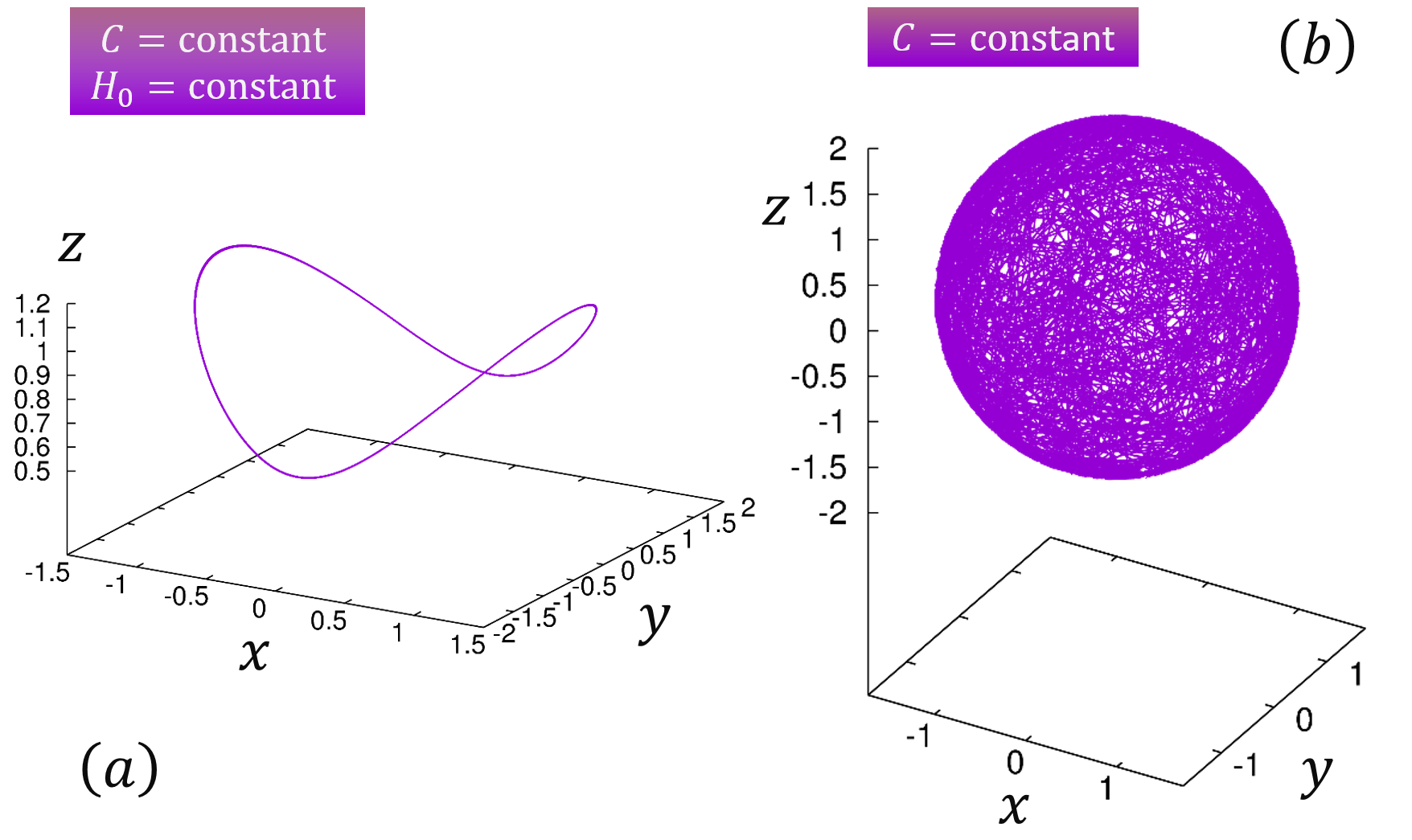}
\caption{\footnotesize (a): Numerical integration of the Euler rotation equation. The orbit is the intersection of the surfaces $C$ and $H_{0}$. (b): Numerical integration of \eqref{X3}.
If the Hamiltonian is perturbed $\nabla H=\nabla H_{0}+\bol{\Gamma}$, the particle explores the surface $C$.}
\label{fig2}
\end{figure}	

Next, consider the antisymmetric operator $\bol{w}=\lr{\cos z-\sin y,-\sin z,\cos y}$ with the same Hamiltonian $H_{0}$. 
One can check that $\JI{w}=\bol{w}^{2}$ so that no Casimir invariant exists. 
The unperturbed orbit is shown in figure \ref{fig3}(a). 
This time the trajectory is spiraling above the energy surface $H_{0}$. 
The absence of an invariant measure is also manifest.
Again, perturb the Hamiltonian as $\nabla H=\nabla H_{0}+\bol{\Gamma}$. 
The resulting orbit is shown in figure \ref{fig3}(b).
Notice that no integral surface exists anymore.

\begin{figure}[h]
\hspace*{-0.4cm}\centering
\includegraphics[scale=0.24]{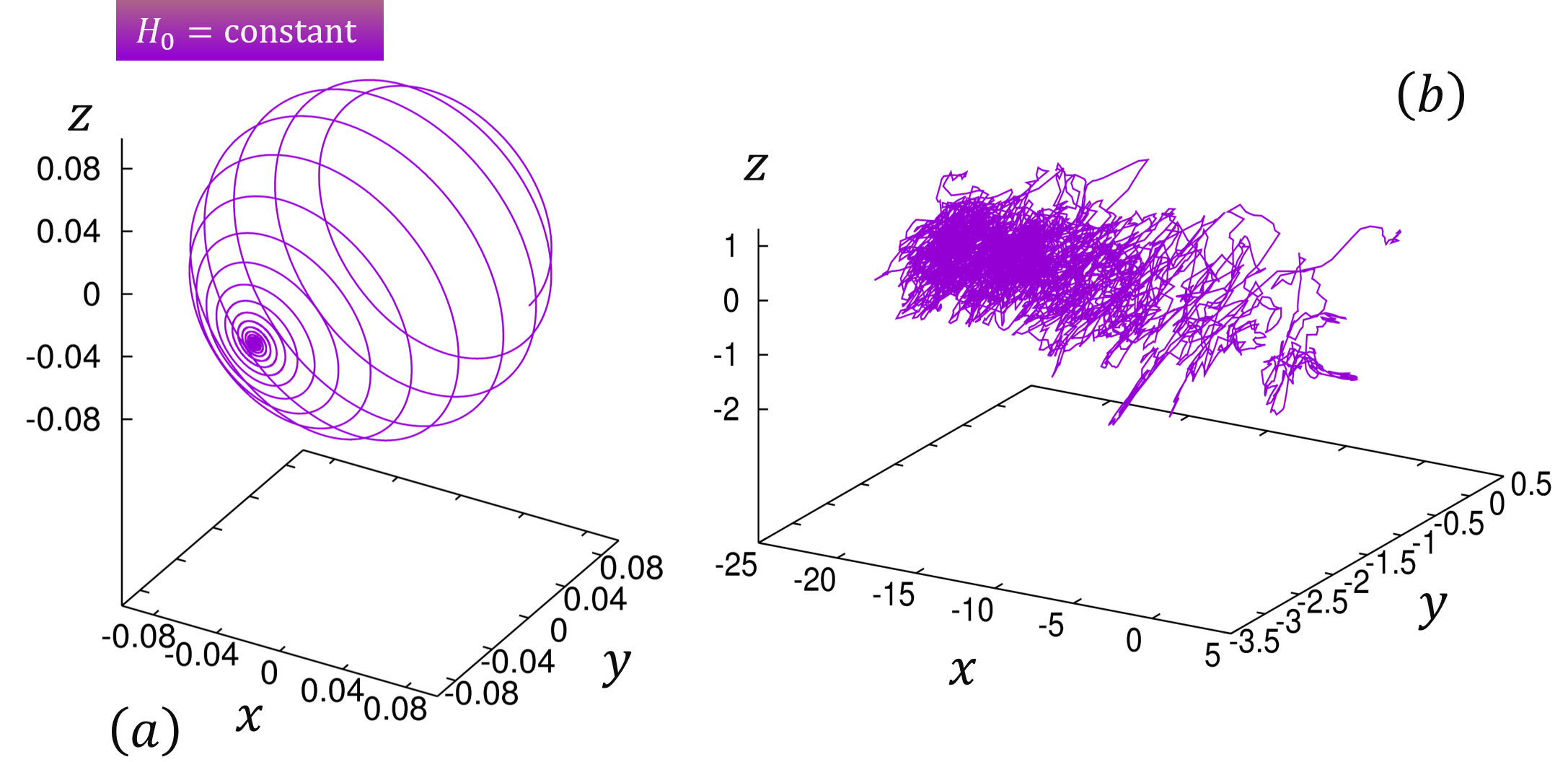}
\caption{\footnotesize Numerical integration of \eqref{EoM3D} for $\bol{w}=\lr{\cos z-\sin y,-\sin z,\cos y}$. (a): The orbit explores the energy surface $H_{0}$ and falls toward a sink. (b): If the Hamiltonian is perturbed $\nabla H=\nabla H_{0}+\bol{\Gamma}$, there are no integral surfaces.}
\label{fig3}
\end{figure}	

In the following part of this section, the analytical solution to the Fokker-Planck equation (\ref{FPEDiff3D_3}) is compared with the numerical integration of the stochastic equation (\ref{XRan2}) for different choices of $\bol{w}$. In each simulation an ensemble of $8\cdot 10^{6}$ particles is considered.
The trajectory of each particle is tracked for the same period of time.
Except when differently specified, the computational domain $\Omega$ is a cube in $\lr{x,y,z}$ space with 
sides of size $6$ and centered at $\bol{x}=\bol{0}$.
The boundary conditions are periodic (except when differently specified) with the period given by the sides of the cube.
The initial condition is a flat (or Gaussian when so specified) probability distribution.
All quantities are given in arbitrary units.

\paragraph{\textit{Uniform operator.}}

The simplest possible situation is given by a uniform operator.
We choose $\bol{w}=\p_{z}$, with $\p_{z}$ the unit vector along the $z$-axis. The helicity density $h=\JI{w}$ identically vanishes because $\nabla\times\bol{w}=\bol{0}$. Therefore, such $\bol{w}$ is a Poisson operator. 
The resulting dynamics $\bol{v}=\p_{z}\times\bol{\Gamma}$ can be thought as the $\bol{E}\cp\bol{B}$ motion
of a charged particle in a constant magnetic field $B=w^{-1}=1$ (remember that in the case of $\bol{E}\cp\bol{B}$
drift  $\bol{w}=\bol{B}/B^{2}$). 
It is also clear that the volume element $dx dy dz$ is an invariant measure for any choice of the Hamiltonian function, and that $\mf{B}=0$.
The analytical form of the equilibrium probability distribution is then determined by observing that $\lambda=1$ and $C=z$. Therefore, in light of \eqref{feq3Dh0}:
\begin{equation}
f^{eq}=\lim_{t\rightarrow\infty}f=A \exp\left\{-\gamma \mc{F}\lr{z}\right\}~in~\Omega,
\end{equation}
Furthermore, since the initial distribution is flat, the diffusion process $\bol{v}=\p_{z}\times\bol{\Gamma}$, which is constrained
in the $\lr{x,y}$ plane, cannot generate any inhomogeneity in the $\p_{z}$ direction.
Hence, $f$ must remain constant throughout the simulation. 
The result of the simulation is shown in figure \ref{fig4}. 

\begin{figure}[H]
\hspace*{-0.3cm}\centering
\includegraphics[scale=0.3]{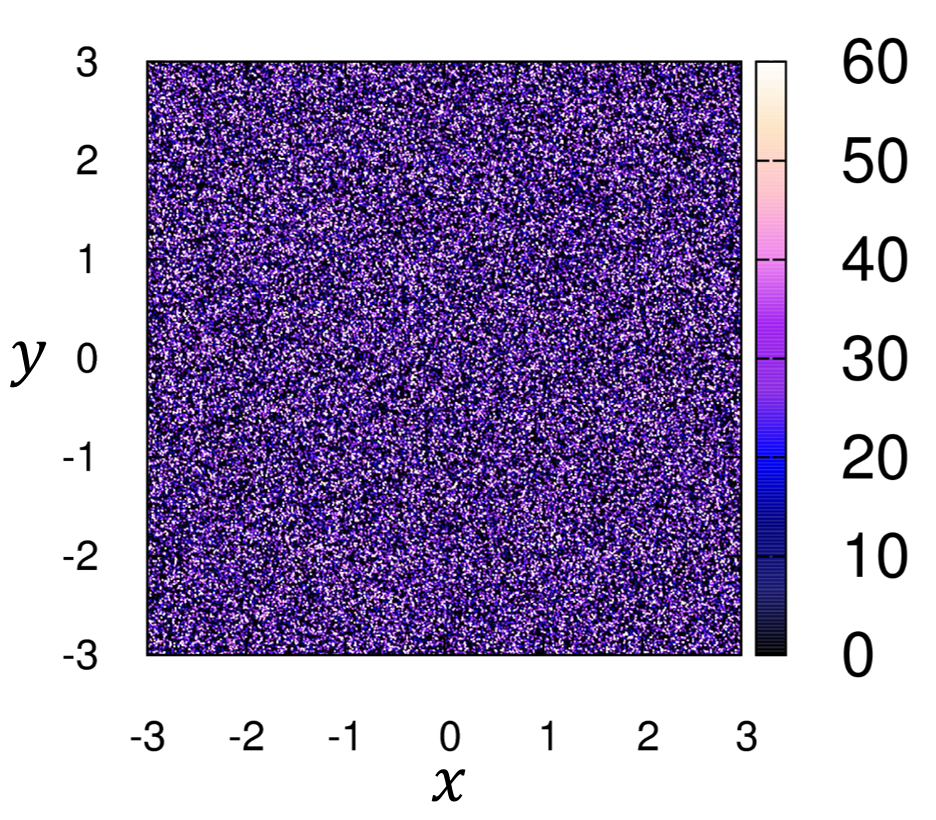}
\caption{\footnotesize Calculated equilibrium probability distribution $f$ in the $\lr{x,y}$ plane at $z=0$ 
with constant Poisson operator $\bol{w}=\p_{z}$.}
\label{fig4}
\end{figure}

\paragraph{\textit{Poisson operator on an invariant measure.}}

Next, we consider the following Poisson operator:
\begin{equation}
\bol{w}=\nabla C=\nabla \lr{z-\cos x-\cos y}.\label{POpIM}
\end{equation}
The Jacobi identity $h=0$ is identically satisfied because $\nabla\times\bol{w}=\nabla\times\nabla C=\bol{0}$, also implying that $\bol{w}$ is a Beltrami operator since $\mf{B}=0$. 
If we interpret the resulting dynamics as the motion of a charged particle in the magnetic field $\bol{B}=\bol{w}/w^{2}$ (given the generality of $\bol{w}$, we do not require $\nabla\cdot\bol{B}=0$ in these examples),
the magnetic field strength is:
\begin{equation}
B=\lr{1+\sin{x}^{2}+\sin{y}^{2}}^{-1/2}.\label{B9_2}
\end{equation}
See figure \ref{fig5}(a) for the plot of $B$. This time the Casimir invariant whose gradient spans the kernel of $\bol{w}$ is the function $C=z-\cos{x}-\cos{y}$.
Using (\ref{eq3}), we also know that $dx dy dz$ is an invariant measure for any choice of the Hamiltonian function.
In light of \eqref{feq3Dh0}, we expect the equilibrium probability distribution to be:
\begin{equation}
f^{eq}=\lim_{t\rightarrow\infty}f=A\exp\{-\gamma\mc{F}\lr{C}\}~in~\Omega.
\end{equation}
Let $f_{0}=f\lr{t=0}$ be the (constant) value of the probability distribution at $t=0$ and $\Omega=\left[-\Delta x/2,\Delta x/2\right]\times\left[-\Delta y/2,\Delta y /2\right]\times\left[-\Delta z/2,\Delta z/2\right]$ the computational domain. 
Since the diffusion process cannot redistribute particles among different levels sets of $C$, 
the number of particles $dN$ on each level set must be preserved, implying $dN\lr{t=0}=f_{0}dC\int{dx\w dy}=f_{0}\Delta x\,\Delta y\, dC=dN\lr{t\rightarrow\infty}=f^{eq} dC \int{ dx\w dy}=f^{eq}\Delta x\,\Delta y\, dC$. But then $f^{eq}=f_{0}=$constant. 
Therefore, the distribution $f$ must remain constant throughout the simulation.
Figure \ref{fig5}(b) shows the results of the numerical simulation.
In particular, notice that the distribution remains flat regardless of the fact that
the random process is spatially inhomogeneous. 

\begin{figure}[H]
\hspace*{-0.1cm}\centering
\includegraphics[scale=0.27]{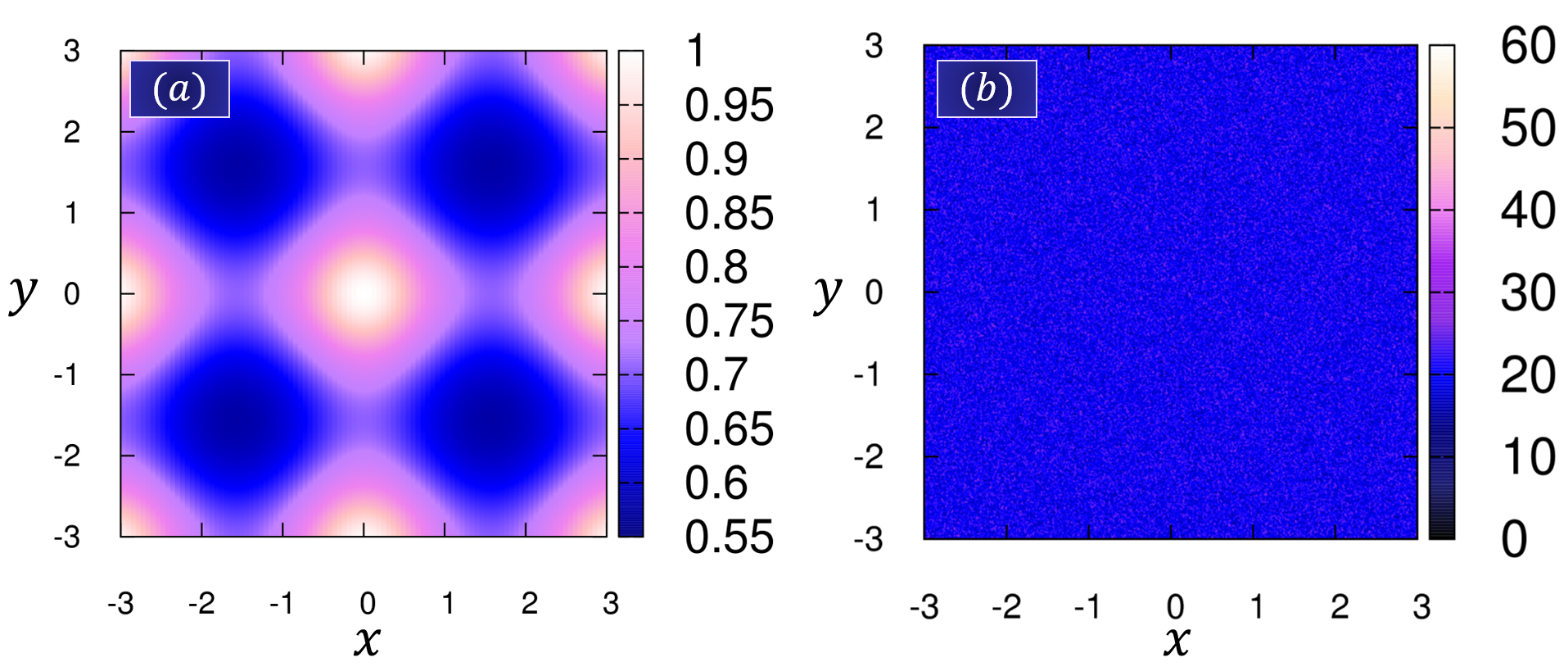}
\caption{\footnotesize (a): Magnetic field strength \eqref{B9_2} in the $\lr{x,y}$ plane. (b) Calculated equilibrium probability distribution $f$ in the $\lr{x,y}$ plane at $z=0$ 
with Poisson operator \eqref{POpIM}.}
\label{fig5}
\end{figure}

\paragraph{\textit{Poisson operator in arbitrary coordinates.}}

Consider now the Poisson operator:
\begin{equation}
\bol{w}=\lambda \nabla C=\lr{\sqrt{1+\cos{x}^{2}}}\nabla\lr{z-\cos{x}-\cos{y}}.\label{POpNoIM}
\end{equation}
Here $\lambda=\sqrt{1+\cos{x}^{2}}\neq 0$ and $C=z-\cos{x}-\cos{y}$. 
The Jacobi identity is easily verified, $h=\lambda\nabla C\cdot\nabla\cp \lambda\nabla C=0$,
and $C$ is a Casimir invariant.
The corresponding magnetic field strength:
\begin{equation}
B=\left[\lr{1+\cos^{2}{x}}\lr{1+\sin^{2}{x}+\sin^{2}{y}}\right]^{-1/2},\label{B9_3}
\end{equation}
is shown in figure \ref{fig6}(a). According to (\ref{eq3}), this time the invariant measure is given by the volume element $\lambda^{-1}dx dy dz$. 
In light of \eqref{feq3Dh0}, 
we expect the solution to converge to a profile of the type $f\propto \lambda^{-1}$.
Figure \ref{fig6}(b) shows a density plot of $\lambda^{-1}$.
Figure \ref{fig6}(c) shows the result of the numerical simulation.
\begin{figure}[H]
\hspace*{-0.1cm}\centering
\includegraphics[scale=0.27]{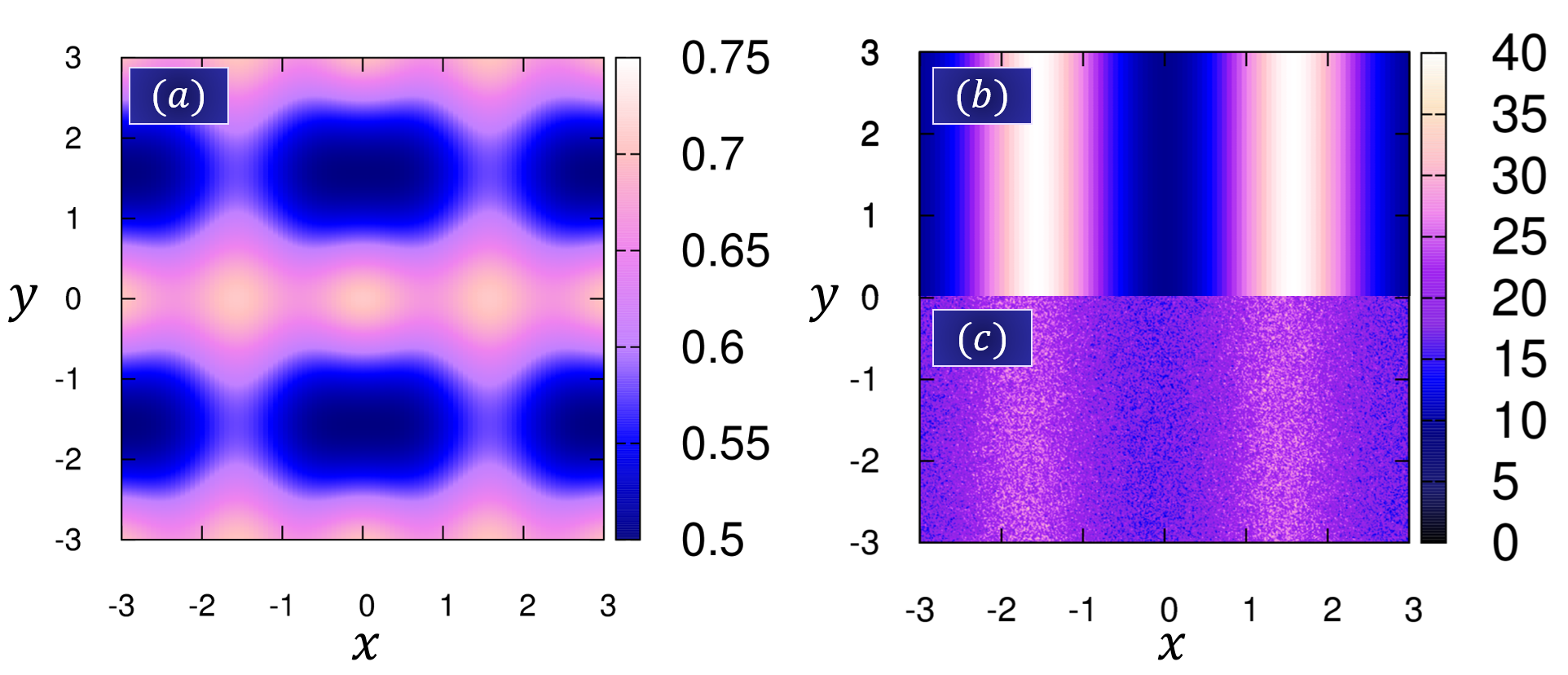}
\caption{\footnotesize (a): Magnetic field strength \eqref{B9_3} in the $\lr{x,y}$ plane. (b): Spatial profile of $\lambda^{-1}$ in the $\lr{x,y}$ plane. (c): Calculated equilibrium probability distribution $f$ in the $\lr{x,y}$ plane at $z=0$ with Poisson operator \eqref{POpNoIM}. The scale at the right of (b) and (c) refers to plot (c).}
\label{fig6}
\end{figure}

\paragraph{\textit{Beltrami operator.}}
Next, consider the operator:
\begin{equation}
\bol{w}=\lr{\cos z+\sin z}\p_{x}+\lr{\cos z-\sin z}\p_{y}.\label{wB}
\end{equation}
One can verify that $h=\bol{w}^{2}=2\neq 0$.
Therefore, $\bol{w}$ is not a Poisson operator.
Furthermore, the field force is $\bol{b}=\bol{w}\cp\bol{w}=\bol{0}$.
This means that $\bol{w}$ is a Beltrami operator.
The corresponding magnetic field strength is constant: $B=w^{-1}=1/\sqrt{2}$.
By theorem \ref{thmIII.1}, $\nabla f=\bol{0}$ in $\Omega$ when $t\rightarrow\infty$. 
This is confirmed by the simulation, figure \ref{fig7}. 
 
\begin{figure}[H]
\hspace*{-0.4cm}\centering
\includegraphics[scale=0.3]{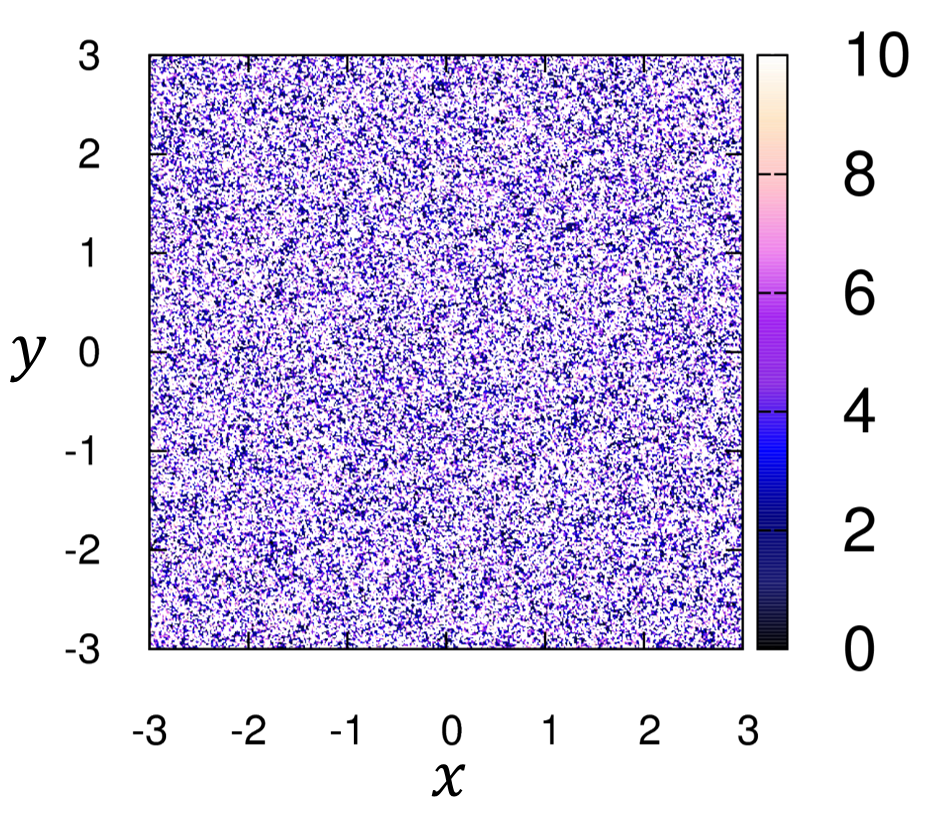}
\caption{\footnotesize Calculated equilibrium probability distribution $f$ in the $\lr{x,y}$ plane at $z=0$ with Beltrami operator \eqref{wB}.}
\label{fig7}
\end{figure}

\paragraph{\textit{Antisymmetric operator.}}
Consider the operator:
\begin{equation}
\bol{w}=\p_{x}+\lr{\sin{x}+\cos{y}}\p_{y}+\lr{\cos{x}}\p_{z}.\label{APOp1}
\end{equation}
The helicity density is $h=1+\sin{x}\cos{y}\geq 0$, 
meaning that the Jacobi identity is violated almost everywhere. 
Furthermore, the field charge is given by $\mf{B}=-4\sin{x}\cos{y}$, which is finite except in a set of measure zero. 
Therefore, this operator is neither a Poisson operator, nor a Beltrami operator in the chosen coordinate system. The corresponding magnetic field strength is:
\begin{equation}
B=w^{-1}=\left[1+\lr{\sin{x}+\cos{y}}^{2}+\cos^{2}{x}\right]^{-1/2}.\label{B9_5}
\end{equation}
A density plot of $B$ is given in figure \ref{fig8}(a).
The result of the corresponding numerical simulation is given in figure \ref{fig8}(b).
Notice that there is a similarity between the profile of magnetic field strength $B=w^{-1}$ and that of the equilibrium probability distribution $f$. This is in agreement with the behavior $f^{eq}\propto Be^{-\zeta}$ obtained in equation \eqref{feqbdz} for the special case $\hb{b}=\nabla\zeta$. 

\begin{figure}[H]
\hspace*{-0.1cm}\centering
\includegraphics[scale=0.44]{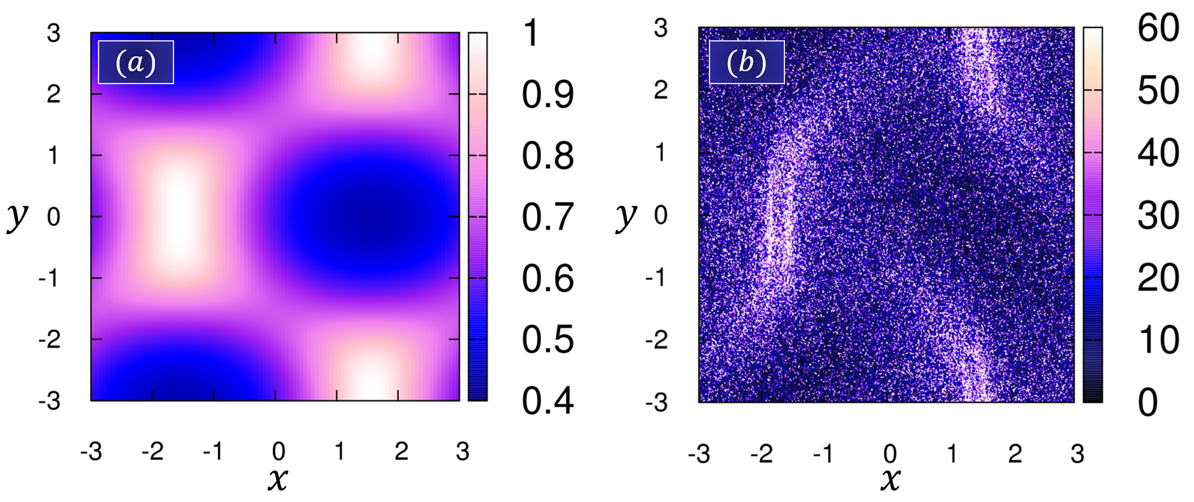}
\caption{\footnotesize (a): Magnetic field strength \eqref{B9_5} in the $\lr{x,y}$ plane. (b): Calculated equilibrium probability distribution $f$ in the $\lr{x,y}$ plane at $z=0$ 
with antisymmetric operator \eqref{APOp1}.}
\label{fig8}
\end{figure}

\paragraph{\textit{Antisymmetric operator with unit norm.}}\label{APUnitNorm}

In the previous paragraph, we analyzed an antisymmetric operator and 
observed that the profile of the probability distribution resembled that
of the magnetic field strength $B=w^{-1}$. To understand the role of the field charge in determining the probability distribution, we consider the antisymmetric operator:
\begin{equation}
\hb{w}=\frac{1}{\sqrt{1+\cos^{2}{x}}}\lr{\cos{y},\cos{x},\sin{y}}.\label{ApOp3}
\end{equation} 
Observe that $B=\hat{w}^{-1}=1$ (and thus $\bol{B}=\hb{w}$).
One can check that the Jacobi identity is not satisfied and thus $\hb{w}$ is not a Poisson operator.
The field charge $\hat{\mf{B}}$ of the operator $\hb{w}$ does not vanish (the lengthy expression of $\hat{\mf{B}}$ is omitted).
Therefore, $\hb{w}$ is not a Beltrami operator in the chosen coordinate system.

The density profile obtained from the numerical simulation is shown in figure \ref{fig9}(b).
Regardless of the fact that $B=\hat{w}^{-1}=1$, an heterogeneous structure is self-organized.
The determinant of this structure is the non-vanishing field charge $\hat{\mf{B}}$. In fact, there is a strong similarity between the profile of the probability distribution and that of $\hat{\mf{B}}$ (compare figure \ref{fig9}(b) with figure \ref{fig9}(a)).

\begin{figure}[H]
\hspace*{-0.1cm}\centering
\includegraphics[scale=0.42]{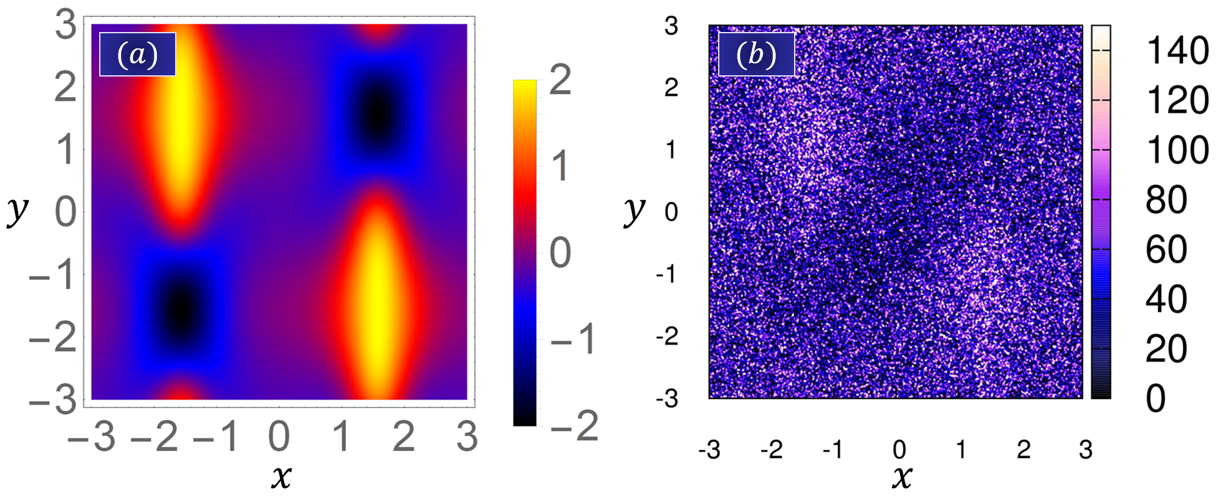}
\caption{\footnotesize (a): Plot of $\hat{\mf{B}}$ for $\hb{w}$ given by equation (\ref{ApOp3}). (b): Calculated equilibrium probability distribution $f$ in the $\lr{x,y}$ plane at $z=0$ with antisymmetric operator \eqref{ApOp3}.}
\label{fig9}
\end{figure}

\paragraph{\textit{The Landau-Lifshitz Equation.}}

The last case we consider is the Landau-Lifshitz equation describing the time evolution of the magnetization $\bol{x}$ in a ferromagnet (specifically, we study equation (35) of \cite{Landau_2}). 
Without entering into details, the Hamiltonian of the system, physically corresponding to the
total magnetization, is given by $H_{0}=\bol{x}^{2}/2$.   
Therefore, in this simulation the perturbed Hamiltonian $H$ is such that $\nabla H=\nabla H_{0}+\bol{\Gamma}$. The relevant operator is:
\begin{equation}
\bol{w}=\gamma\mc{H}-\frac{\sigma}{x^{2}}\mc{H}\times\bol{x}.\label{APOp4}
\end{equation}
Here, $\gamma$ is the so called damping parameter, $\sigma$ a physical constant, and $\mc{H}$ the effective magnetic field.
The effective magnetic field $\mc{H}$ is chosen to be $\mc{H}=\lr{c,0,z}$, 
where $c$ represents a constant external magnetic field. Then, equation \eqref{APOp4} can be rewritten as:
\begin{equation}
\bol{w}=\lr{c+\sigma \frac{zy}{\bol{x}^{2}}}\p_{x}+\sigma\frac{z\lr{c-x}}{\bol{x}^{2}}\p_{y}+\lr{z-\sigma\frac{cy}{\bol{x}^{2}}}\p_{z}.\label{APOp4_2}
\end{equation}
One can verify that this operator violates the Jacobi identity and that the field charge
does not vanish. 
Therefore, $\bol{w}$ is not a Poisson operator, nor a Beltrami operator.
In figure \ref{fig10} the results of the numerical simulation are shown.
This time, the initial condition is a Maxwell-Boltzmann distribution centered at $\bol{x}=\lr{0,0,z_{0}}$.
Furthermore, the trajectory of each magnetization is followed as far as it goes, i.e. no boundary conditions are used.
Notice how the probability distribution becomes strongly anisotropic,
with preferential alignment of the magnetization along the $z$-axis (representing the direction of easiest
magnetization of the ferromagnetic crystal).

\begin{figure}[h]
\hspace*{-0cm}\centering
\includegraphics[scale=0.3]{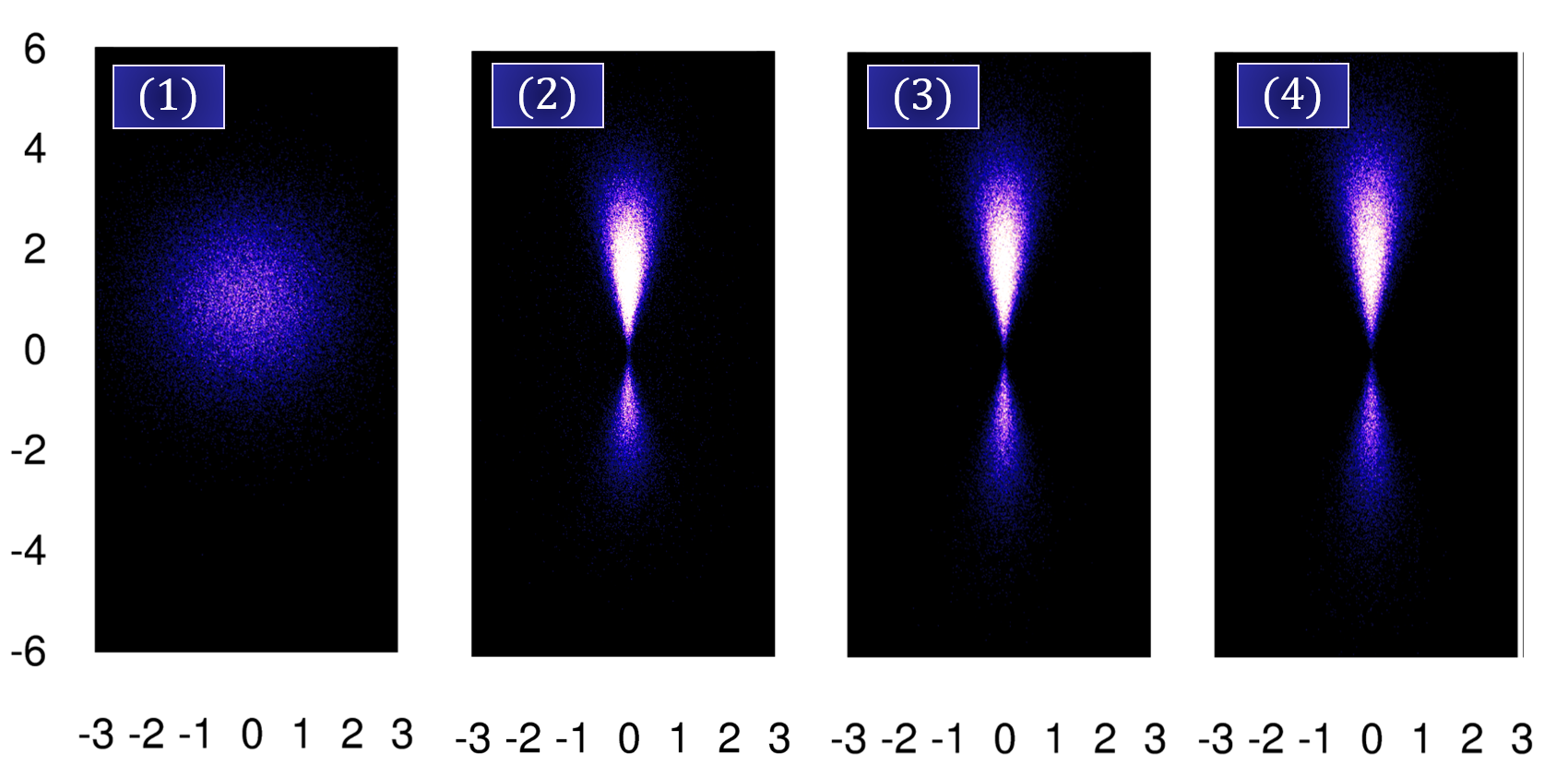}
\caption{\footnotesize Time evolution of the probability distribution $f$ in the $\lr{x,z}$ plane at $y=0$.
Each plot number $i$ corresponds to the instant $t=i \Delta t$, where $\Delta t$ is a fixed time interval.}
\label{fig10}
\end{figure}


\section{Conservative Dynamics and Topological Constraints}

The remaining part of this paper is devoted to the generalization of the theory to $n$ dimensions.
The key idea is that, by invoking change of coordinates, the classification of $\bol{w}$ in terms
of $h$ and $\mf{B}$ developed in section II can be generalized to include operators 
that satisfy the criterion $\mf{B}=0$ in different reference frames.
We will see that in this way the Poisson operator of Hamiltonian systems 
defines a subclass of Beltrami operators.  
This requires a coordinate free formulation. 
For this reason, the formalism
of differential geometry will be used.

In this section we review the concepts of antisymmetric operators, Poisson operators, and topological constraints, and introduce the mathematical notation used in the rest of the paper. 

Let $\mathcal{M}$ be a smooth manifold of dimension $n$.
An \textit{antisymmetric operator} is a bivector field $\mathcal{J}\in \bigwedge^{2}T\mathcal{M}$. 
Let $\left(x^{1},...,x^{n}\right)$ be a coordinate system on $\mathcal{M}$. Consider the
tangent basis $\left(\partial_{1},...,\partial_{n}\right)$. We have:
\begin{equation}
\mathcal{J}=\sum_{i<j}\mathcal{J}^{ij}\partial_{i}\wedge\partial_{j}=\frac{1}{2}\mathcal{J}^{ij}\partial_{i}\wedge\partial_{j},\,\,\,\,\,\,\,\,\,\mathcal{J}^{ij}=-\mathcal{J}^{ji}.
\end{equation}
Here and throughout this study 
we shall assume $\mJ^{ij}\in C^{\infty}\lr{\mc{M}}$, except when differently specified.
The matrix $\mathcal{J}^{ij}$ is antisymmetric and defines an antisymmetric bilinear inner product on pairs of functions $f,g\in C^{\infty}\lr{\mc{M}}$ 
called antisymmetric bracket:
\begin{equation}
\left\{f,g\right\}=\mc{J}\lr{df,dg}=-\mc{J}\lr{dg,df}=\mc{J}^{ij}f_{i}g_{j}.
\end{equation}
In this notation, lower indices applied to a function indicate derivation, i.e. $f_{i}=\p f/\p x^{i}$.

An antisymmetric operator $\mathcal{J}\in\bigwedge^{2}T\mathcal{M}$ 
and an Hamiltonian function $H\in C^{\infty}\left(\mathcal{M}\right)$
define a \textit{conservative vector field} $X\in T\mathcal{M}$ as:
\begin{equation}
X=\mathcal{J}\left(dH\right)=\mathcal{J}^{ij}H_{j}\partial_{i}.\label{AH}
\end{equation}
\noindent For the $3$D case, one can verify that by setting:
\begin{dmath}
\mathcal{J}=\mJ^{zy}\p_{z}\w\p_{y}+\mJ^{xz}\p_{x}\w\p_{z}+\mJ^{yx}\p_{y}\w\p_{x}
=w_{x}\partial_{z}\wedge\partial_{y}+w_{y}\partial_{x}\wedge\partial_{z}+w_{z}\partial_{y}\wedge\partial_{x},\label{J3D_}
\end{dmath}
we have in a unique manner $X=\mc{J}\lr{dH}=\bol{w}\times\nabla H$.
Thanks to antisymmetry, a conservative vector field $X$ always preserves the Hamiltonian $H$ along the flow. 

The antisymmetric bracket defined by $\mathcal{J}$ is called a Poisson bracket if it satisfies the Jacobi identity:
\begin{equation}
h=\mathcal{J}^{im}\frac{\partial\mathcal{J}^{jk}}{\partial x^{m}}+\mathcal{J}^{jm}\frac{\partial\mathcal{J}^{ki}}{\partial x^{m}}+\mathcal{J}^{km}\frac{\partial\mathcal{J}^{ij}}{\partial x^{m}}=0,\label{JIJ}
\end{equation}
$\forall i,j,k=1,...,n$. In this case, $\mathcal{J}$ is called a \textit{Poisson operator} and the associated vector field $X$ a \textit{non-canonical Hamiltonian vector field}.

If $\mathcal{J}$ is invertible (and therefore $n=2m$, $m\in\mathbb{N}$) with inverse $\omega\in \bigwedge^{2}T^{\ast}\mc{M}$, the Jacobi identity is equivalent to demanding that $d\omega=0$ (remember that a dual definition of Hamiltonian system can be given in terms of the symplectic $2$-form $\omega$ as $i_{X}\omega=-dH$, with $d\omega=0$). Canonical Hamiltonian systems correspond to a special class of Poisson operators
called symplectic operators (or simplectic matrices):
\begin{equation}
\mathcal{J}_{c}=\sum_{i=1}^{m}\partial_{m+i}\wedge\partial_{i},\label{SymOp}
\end{equation}
The vector field $X=\mathcal{J}_{c}\left(dH\right)$ is then called a canonical Hamiltonian vector field.
In light of Darboux's theorem \cite{DeLeon_3,Arnold_4}, given a constant rank Poisson operator
$\mc{J}$ of dimension $n=2m+r$ ($2m$ is the rank), one can always find a local coordinate change by which $\mc{J}$ is expressed in the form (\ref{SymOp}).

In general, an antisymmetric operator $\mJ$ needs not to be invertible, i.e. its rank
can be smaller than its dimension, $rank\lr{\mJ}\leq dim\lr{\mJ}$. When this happens, $\mJ$ has a non-trivial kernel, $ker\lr{\mJ}=\left\{\theta\in T^{\ast}\mc{M}:\mc{J}\lr{\theta}=0\right\}$.
Clearly, we must have $dim\lr{\mJ}=rank\lr{\mJ}+dim\lr{ker\lr{\mJ}}$. 
Notice that any $1$-form $\theta\in ker\left(\mathcal{J}\right)$ is orthogonal to 
the conservative vector field $X=\mathcal{J}\left(dH\right)$ for any choice of $H$:
\begin{equation}
\theta\lr{X}=i_{X}\theta=\theta_{i}\mathcal{J}^{ij}H_{j}=0\,\,\,\,\,\,\,\,\forall H.
\end{equation}
This condition represents a geometrical constraint that is independent of the properties of matter (which are
encoded in $H$), i.e. it defines a \textit{topological constraint}. 
A collection of $r$ constraints on a $2m+r$ dimensional manifold $\mc{M}^{2m+r}$ defines a $2n$ dimensional distribution $\Delta_{2m}=\left\{X\in T\mc{M}^{2m+r}:\theta_{i}\lr{X}=0\,\,\,\,\forall i=1,...,r\right\}$. As a consequence of Darboux's theorem, the distribution $\Delta_{2m}$ associated to a Poisson operator $\mc{J}$ of dimension $2m+r$, with $dim\lr{ker\lr{\mJ}}=r$, is always integrable in the sense of Frobenius' theorem \cite{Frankel,Agricola,Kobayashi}, i.e. there exists $r$ scalar functions $C^{i}$ called Casimir invariants such that $\mJ\lr{dC^{i}}=0$ and therefore $\Delta_{2m}=\left\{X\in T\mc{M}^{2m+r}:dC^{i}\lr{X}=0\,\,\,\,\forall i=1,...,r\right\}$. The word `invariant' refers to the fact that the $C^{i}$s are constants of motion that do not depend on the specific choice of $H$. 


\section{Geometrical Classification of Antisymmetric Operators}

The objective of this section is to produce a geometrical classification of
antisymmetric operators that is relevant from the standpoint of statistical mechanics. 
For this purpose we need a representation of antisymmetric operators in terms of differential forms. 

Let $\mc{J}\in\bigwedge^{2}T\mc{M}$ be an antisymmetric operator.
Let $vol^{n}=g dx^{1}\w ... \w dx^{n}$ be a volume element on $\mc{M}$, with $g\neq 0$ and $g\in C^{\infty}\lr{\mc{M}}$.
The \textit{covorticity} $n-2$ form with respect to $vol^{n}$ is defined as:
\begin{dmath}
\mc{J}^{n-2}=i_{\mc{J}}vol^{n}=
\sum_{i<j}\left(-1\right)^{i+j-1}g\mathcal{J}^{ij}
\left(i_{\partial_{i}\wedge\partial_{j}}dx^{i}\wedge dx^{j}\right)\wedge dx^{n-2}_{ij}=
2\sum_{i<j}\left(-1\right)^{i+j-1}g\mathcal{J}^{ij}dx_{ij}^{n-2}.\label{Jn-2}
\end{dmath}
\noindent In this notation $dx_{ij}^{n-2}=dx^{1}\w ... \w dx^{i-1}\w dx^{i+1}\w ... \w dx^{j-1}\w dx^{j+1}\w ... \w dx^{n}$.
Next, it is useful to define the \textit{cocurrent} $n-1$ form of $\mc{J}$ with respect to the volume form $vol^{n}$ on $\mc{M}$ as: 
\begin{equation}
\mc{O}^{n-1}=d\mc{J}^{n-2}.\label{Vorticity}
\end{equation}

In the same way the closeness of the $2$-form $\omega$ defines Hamiltonian mechanics,
the closeness of the $n-2$ form $\mc{J}^{n-2}$ is a powerful condition. 
Indeed, we can show that:\\ 

\textit{The conservative vector field $X=\mc{J}\lr{dH}$ admits an invariant measure 
$vol^{n}$ for any choice of the Hamiltonian $H$ if and only if $\mc{O}^{n-1}=0$ on the volume form $vol^{n}$:
\begin{equation}
\mf{L}_{X}vol^{n}=0~\forall H~\iff~\mc{O}^{n-1}=0~on~vol^{n}.\label{EIM}
\end{equation}}

\noindent To see this, note that from \eqref{Vorticity} we have:
\begin{dmath}
\mc{O}^{n-1}=
2\left(-1\right)^{j}\frac{\partial \left(g\mathcal{J}^{ij}\right)}{\partial x^{i}}dx^{n-1}_{j}.\label{dJn-2}
\end{dmath} 
On the other hand:
\begin{equation}
\mf{L}_{X}vol^{n}
=\frac{\partial\lr{g\mc{J}^{ij}}}{\partial x^{i}}H_{j}dx^{1}\w ... \w dx^{n}.\label{LXvoln}
\end{equation}
Hence, \eqref{LXvoln} vanishes for any $H$ if and only if $\mc{O}^{n-1}=0$. 


\subsection{The Measure Preserving Operator}

Equation \eqref{EIM} introduces a notion of invariant measure
that does not depend on the specific choice of the Hamiltonian $H$, but only on the geometrical properties of the operator $\mc{J}$.
To know whether a certain operator $\mc{J}$ admits this kind of Hamiltonian-independent invariant measure
it is therefore sufficient to determine whether a metric $g$ can be found such that $\mc{O}^{n-1}=0$. 

It is now natural to define the \textit{measure preserving operator}:  
an antisymmetric operator $\mc{J}\in\bigwedge^{2}T\mc{M}$ 
will be called measure preserving if there exists a volume form $vol^{n}$ on $\mc{M}$ such that $\mc{O}^{n-1}=0$. 
Evidently, an antisymmetric operator can be measure preserving without satisfying the Jacobi identity \eqref{JIJ}, i.e. without being a Poisson operator.
Furthermore, notice that 
a constant rank Poisson operator is measure preserving.
The proof of this statement, which is omitted, can be obtained by applying Darboux's theorem.

In the next part of the present study it will be shown that on the invariant measure defined by a measure preserving operator the standard results of statistical mechanics can be recovered. Because of the special properties of the measure preserving operator, it is useful to determine whether a general antisymmetric operator can be transformed to a measure preserving form. 
On this regard, the following extension method applies:\\ 

\textit{Let $\mathcal{J}\in\bigwedge^{2}T\mc{M}$ be an antisymmetric operator on a smooth manifold $\mc{M}$ of dimension $n$.
Let $x^{n+1}$ be a new variable with domain $\mc{D}\subset\mathbb{R}$. Then, the $n+1$ dimensional antisymmetric operator on $\bigwedge^{2}T\lr{\mc{M}\times\mc{D}}$:
\begin{equation}
\mathfrak{J}=\mathcal{J}+x^{n+1}\frac{\partial\mathcal{J}^{ij}}{\partial x^{i}}\partial_{j}\wedge\partial_{n+1},\label{ExtJ}
\end{equation}
\noindent is measure preserving.}\\

To prove the statement, it is sufficient to show that on the volume form $vol^{n+1}=dx^{1}\w ... \w dx^{n}\w dx^{n+1}$, the cocurrent $\mc{O}^{n}=d\mf{J}^{n-1}$ vanishes.
Recalling the condition given by equation \eqref{dJn-2}, it follows that:
\begin{equation}
\sum_{i=1}^{n+1}\frac{\partial\mathfrak{J}^{ij}}{\partial x^{i}}=
\frac{\partial\mathfrak{J}^{n+1,j}}{\partial x^{n+1}}+\sum_{i=1}^{n}\frac{\partial\mathfrak{J}^{ij}}{\partial x^{i}}=
x^{n+1}\sum_{i,k=1}^{n}\frac{\partial^{2}\mathcal{J}^{ki}}{\partial x^{i}\partial x^{k}}=0,
\end{equation}
as desired. Finally, observe that the extended system $X^{n+1}=\mf{J}\lr{dH}$ preserves
the form of the original equations of motion $X^{n}=\mc{J}\lr{dH}$ for the original $n$ variables
because the Hamiltonian $H$ does not depend on the new variable $x^{n+1}$, i.e. $H_{n+1}=0$.

\subsection{The Beltrami Operator}\label{BOp}

The remaining task is the generalization of the concept of field charge to arbitrary dimensions $n$. By consistency with equation \eqref{Bness3D}, the field charge of a general antisymmetric operator $\mJ$ must be a $0$-form. Furthermore, since $\mf{B}$ is
the divergence of the vector $\bol{b}$, the generalization of $\bol{b}$ must be an $n-1$ form. 
Hence, it is natural to define the \textit{field force} $n-1$ form of $\mathcal{J}$ as:
\begin{dmath}
b^{n-1}=\mathcal{J}^{n-2}\wedge\ast d\mathcal{J}^{n-2}=4\sum_{i<j}\lr{-1}^{i+j+k-1}g\mJ^{ij}\frac{\p\lr{g\mJ^{lk}}}{\p x^{l}}dx^{n-2}_{ij}\w\ast dx^{n-1}_{k}.\label{bn-1}
\end{dmath}
Then, the \textit{field charge} of $\mathcal{J}$ will be:
\begin{equation}
\mf{B}=\ast d b^{n-1}=
4\frac{\p}{\p x^{i}}\lr{g\mJ^{ij}\frac{\p\lr{g\mJ^{lj}}}{\p x^{l}}}.\label{b0}
\end{equation}
One can check that these definitions correctly reproduce those of the case $n=3$ of $\mathbb{R}^{3}$.


Now we can introduce the notion of Beltrami operator: let $\mc{J}$ be
an antisymmetric operator. If a volume form $vol^{n}=g dx^{1}\w ... \w dx^{n}$ 
can be found such that the field charge is zero, i.e. $\mf{B}=\ast db^{n-1}=0$, 
$\mJ$ is called a \textit{Beltrami operator} on $vol^{n}$.
If the field force $n-1$ form is zero, i.e. $b^{n-1}=0$, $\mJ$ is called a \textit{strong Beltrami operator} on $vol^{n}$.

Suppose that $\mathcal{J}$ is a measure preserving operator with invariant measure $vol^{n}$.
Evidently, such $\mc{J}$ is a strong Beltrami operator on the invariant measure, i.e. $b^{n-1}=0$ on $vol^{n}$. 
This is because a measure preserving operator satisfies $d\mc{J}^{n-2}=0$ on the metric of the invariant measure (recall equation \ref{EIM}). Therefore, the corresponding field force $n-1$ form $b^{n-1}=\mJ^{n-2}\w \ast d\mJ^{n-2}$ identically vanishes.

A sufficient condition for $\mJ$ to be a Beltrami operator is the vanishing of the $i$th component 
$g\mJ^{ij}\frac{\p\lr{g\mJ^{lj}}}{\p x^{l}}$ of the quantity appearing in equation (\ref{b0}).
For the case $n=3$ of $\mathbb{R}^{3}$, this is exactly the requirement $\bol{b}=\bol{0}$.
If the operator $\mJ$ is invertible with inverse $\omega$, such condition degenerates to the
definition of measure preserving operator $\p_{i}\lr{g\mc{J}^{ij}}=0$, and can be cast in a 
metric independent fashion. First, multiply by $\omega^{km}\omega^{mi}$:
\begin{equation}
\omega^{km}\omega^{mi}g\mJ^{ij}\frac{\p\lr{g\mJ^{lj}}}{\p x^{l}}=
g^{2}\left[\omega^{km}\frac{\p\mJ^{lm}}{\p x^{l}}-\frac{\p\log{g}}{\p x^{k}}\right]=0.\label{BI}
\end{equation}
Define the $1$-form $\mf{A}=\omega^{km}\frac{\p\mJ^{lm}}{\p x^{l}}dx^{k}$.
Then, equation (\ref{BI}) reads as $\mf{A}=d\log{g}$. 
If $\Omega$ is an open ball of $\mathbb{R}^{n}$ or a star-shaped open set about $\bol{0}$, 
Poincar\'e's lemma applies, and
equation (\ref{BI}) can be satisfied by demanding that $d\mf{A}=0$, or explicitly:
\begin{equation}
\lr{\frac{\p\omega^{km}}{\p x^{n}}-\frac{\p\omega^{nm}}{\p x^{k}}}\frac{\p\mJ^{lm}}{\p x^{l}}
+\omega^{km}\frac{\p^{2}\mJ^{lm}}{\p x^{l}\p x^{n}}-\omega^{nm}\frac{\p^{2}\mJ^{lm}}{\p x^{l}\p x^{k}}=0.\label{BI2}
\end{equation}
Therefore, by checking the identity (\ref{BI2}) on the domain $\Omega$ above, it is possible to establish whether there exists a coordinate system where an invertible operator is measure preserving. 


Figure \ref{fig11}(a) summarizes the geometrical categorization of antisymmetric operators developed in the present section. Figure \ref{fig11}(b) shows a similar summary for the special and instructive case $n=3$.


\begin{figure}[h]
\centering
\includegraphics[scale=0.21]{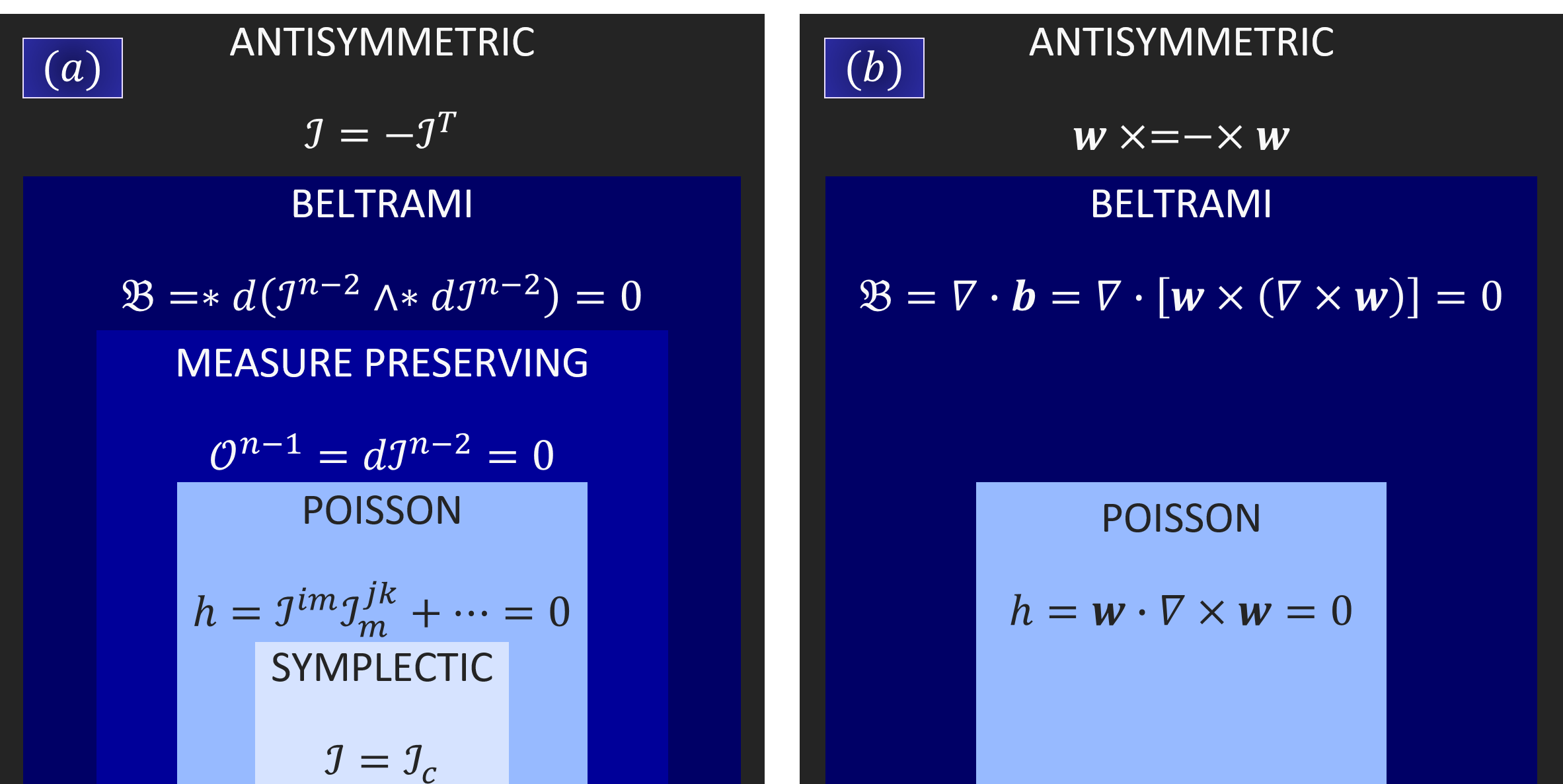}
\caption{\footnotesize (a): The hierarchical structure of antisymmetric operators. Each box is named by the corresponding operator. (b): The hierarchical structure of antisymmetric operators for $n=3$. Notice that measure preserving operators do not appear because they degenerate to Poisson operators when $n=3$. Specifically, the measure preserving condition $\nabla\times\lr{g\bol{w}}=\bol{0}$ reduce to the integrability condition for $\bol{w}$, see (\ref{eq3}). Similarly, the symplectic operator does not appear because canonical pairs cannot be defined in odd dimensions.}
\label{fig11}
\end{figure}

\section{Fokker-Planck Equation}

Consider now an ensemble of particles with an antisymmetric operator $\mc{J}\in\bigwedge^{2}T\mc{M}$ and an Hamiltonian function $H_{0}\in C^{\infty}\lr{\mc{M}}$. In order to construct the evolution equation for the corresponding probability distribution $f$, we must first obtain the stochastic differential equations governing particle dynamics. The motion of a single particle is described by the differential equation:
\begin{equation}
X_{0}=\mJ\lr{dH_{0}}.\label{EoM1}
\end{equation}
First, assume that all the particles in the ensemble are not interacting, each of them obeying equation \eqref{EoM1}. Then, if we switch on some interaction, the energy $H_{0}$ will change according to $H=H_{0}\lr{\bol{x}}+H_{I}\lr{\bol{x},t}$ where $H$ is the new Hamiltonian function accounting for the interaction energy $H_{I}\lr{\bol{x},t}$.
We take $H_{I}$, and thus $H$, to be $C^{\infty}\lr{\mc{M}\times\mathbb{R}_{\geq 0}}$.
The interaction is therefore represented by the vector field $X_{I}$ with components $X_{I}^{i}=\mJ^{ij}H_{Ii}$.
To complete the description of particle dynamics, we further assume that all perturbations 
caused by $H_{I}$ are counterbalanced by a \textit{friction force}:
\begin{equation}
X^{i}_{F}=-\gamma^{ij}H_{0j}=-\frac{1}{2}\beta\mJ^{ik}\mJ^{jk}H_{0j}=\frac{1}{2}\beta\mJ^{ik}X_{0}^{k}.\label{FF}
\end{equation}
Here, $\gamma^{ij}=\frac{1}{2}\beta\mJ^{ik}\mJ^{jk}$ is the friction coefficient with $\beta\in\mathbb{R}$ a spatial constant. 
Since the gradient of the Hamiltonian physically represents force, equation \eqref{FF} leads to a total force $-H_{0i}-H_{Ii}-\frac{1}{2}\beta X_{0}^{i}$ where the friction term is proportional to the velocity as in the usual definition. 

In summary, the equation of motion governing the dynamics of a particle in the ensemble is:
\begin{dmath}
X=X_{0}+X_{I}+X_{F}=
\left[\lr{\mJ^{ij}-\frac{1}{2}\beta\mJ^{ik}\mJ^{jk}}H_{0j}+\mJ^{ij}\Gamma_{j}\right]\p_{i}.\label{EoMN}
\end{dmath}
In the last passage we made the substitution $\mJ^{ij}H_{Ij}=\mJ^{ij}\Gamma_{j}$.
Here, we assumed that the $j$th component of the gradient of $H_{I}$ 
is represented by Gaussian white noise $\Gamma_{j}$, i.e. $H_{Ij}=\Gamma_{j}$.
We will justify this assumption later.

In the following, we will need a slightly more general form of equation \eqref{EoMN}.
Indeed, in equation \eqref{EoMN} white noise is applied in the same coordinate system $\bol{x}=\lr{x^{1},...,x^{n}}$ used to describe the dynamics. However, we want to be able to perturb the ensemble in a different coordinate system, say $\bol{y}=\lr{y^{1},...,y^{n}}$.
Restricting to the cases in which the map $\mc{T}:\bol{x}\rightarrow\bol{y}$ is a diffeomorphism, we introduce the tensor $R^{m}_{j}=\p y^{m}/\p x^{j}$ and
generalize equation \eqref{EoMN}:
\begin{dmath}
X=\left[\lr{\mJ^{ij}-\frac{1}{2}\beta\mJ^{ir}R_{r}^{k}\mJ^{js}R_{s}^{k}}H_{0j}+\mJ^{ij}R_{j}^{r}\Gamma_{r}\right]\p_{i}.\label{EoMN2}
\end{dmath}
Here, the friction coefficient is $\gamma^{ij}=\frac{1}{2}\beta\mJ^{ir}R_{r}^{k}\mJ^{js}R_{s}^{k}$ and 
we used the formula $H_{Ij}=R_{j}^{r}\Gamma_{r}$. Now white noise is applied in the new coordinates $\bol{y}$ since $\p H_{I}/\p y^{r}=\Gamma_{r}$. 

Observe that equation \eqref{EoMN2} is now a stochastic differential equation.
Therefore, by application of the standard procedure (see for example \cite{Gardiner,Risken,Sato1}), we can derive the corresponding Fokker-Planck equation for the probability distribution $f$ on the volume element $vol^{n}=dx^{1}\w ... \w dx^{n}$. We have:

\begin{dmath}
\frac{\partial f}{\partial t}=\frac{\partial}{\partial x^{i}}\left[-\left(\mathcal{J}^{ij}-\frac{1}{2}\beta\mJ^{ir}R_{r}^{k}\mJ^{js}R_{s}^{k}\right)H_{0j}f+
\frac{1}{2}\frac{\partial}{\partial x^{j}}\left(\mathcal{J}^{ir}R^{k}_{r}\mathcal{J}^{js}R^{k}_{s}f\right)-\alpha\frac{\partial\mathcal{J}^{ir}R^{k}_{r}}{\partial x^{j}}\mathcal{J}^{js}R^{k}_{s}f\right].\label{FPE_6}
\end{dmath}

Finally, we must assign a specific value to the parameter $\alpha\in\left[0,1\right]$ (which defines the stochastic integral \cite{Gardiner,Risken,Sato1}), for the stochastic differential equation \eqref{EoMN2} and for the Fokker-Planck equation \eqref{FPE_6} to make mathematically sense.
Assuming that the white noise $\bol{\Gamma}$ appearing in the equations is the limiting representation of a continuous perturbation, we take the value $\alpha=1/2$ (corresponding to the Stratonovich definition of the stochastic integral). When $\alpha=1/2$, equation \eqref{FPE_6} reduces to:
\begin{dmath}
\frac{\partial f}{\partial t}=\frac{\partial}{\partial x^{i}}\left[-\left(\mathcal{J}^{ij}-\frac{1}{2}\beta\mJ^{ir}R_{r}^{k}\mJ^{js}R_{s}^{k}\right)H_{0j}f+
\frac{1}{2}\mathcal{J}^{ir}R^{k}_{r}\frac{\partial}{\partial x^{j}}\left(\mathcal{J}^{js}R^{k}_{s}f\right)\right].\label{FPE_6_2}
\end{dmath}

\noindent Observe that the matrix $R^{k}_{r}$ can be interpreted as the square root of a generalized diffusion parameter.

\section{H-theorem for measure preserving operators}

The derived Fokker-Planck equation \eqref{FPE_6} shows that the behavior of the 
probability distribution $f$ depends on three factors: the energy $H$ representing the properties of matter,
the metric of space characterized by the operator $\mJ$, and the type of perturbations described by the tensor
$R^{k}_{r}$ and the parameter $\alpha$ (notice that physically $R^{k}_{r}$ accounts for the spatial properties and $\alpha$ for the type of time evolution of the perturbations).  
In this section we examine the form of 
$f^{eq}=\lim_{t\rightarrow\infty}f$. 
It is convenient to define the concept of Fokker-Planck velocity $Z$.
Since the probability $fvol^{n}$ enclosed in each volume element must be preserved along the trajectories, if $Z\in T\mc{M}$ is the dynamical flow generating the evolution of such probability,
we must have the following conservation law:
\begin{equation}
\lr{\p_{t}+\mf{L}_{Z}}fvol^{n}=
\left[\frac{\p f}{\p t}+\frac{\p}{\p x^{i}}\lr{fZ^{i}}\right]vol^{n}=0.
\end{equation} 
Comparing this equation with the Fokker-Planck equation \eqref{FPE_6}, wee see that:
\begin{dmath}
Z^{i}=\left(\mathcal{J}^{ij}-\frac{1}{2}\beta\mJ^{ir}R_{r}^{k}\mJ^{js}R_{s}^{k}\right)H_{0j}
-\frac{1}{2f}\frac{\partial}{\partial x^{j}}\left(\mathcal{J}^{ir}R^{k}_{r}\mathcal{J}^{js}R^{k}_{s}f\right)
+\alpha\frac{\partial\mathcal{J}^{ir}R^{k}_{r}}{\partial x^{j}}\mathcal{J}^{js}R^{k}_{s}.\label{Z}
\end{dmath}
The quantity $Z$ is called the Fokker-Planck velocity of the system. 

We anticipated that, in the absence of canonical phase space, the form of $f^{eq}$ departs from the
standard Maxwell-Boltzmann distribution and takes a novel form depending on the operator $\mJ$. 
On this regard, the following convergence theorem for measure preserving operators holds.\\

\textit{Assume the following conditions:\\
$\bullet$ $\mathcal{J}\in \bigwedge^{2}T\mathcal{M}$ is a measure preserving operator of $C^{2}$ class. \\ 
$\bullet$ $\boldsymbol{x}=\left(x^{1},...,x^{n}\right)$ is a 
coordinate system on $\mathcal{M}$ endowed with the invariant measure, 
i.e. $\partial_{i}\mathcal{J}^{ij}=0$ $\forall j=1,...,n$. \\
$\bullet$ Let $W_{i}$, $i=1,...,n$ be $n$ Wiener processes, with $dW_{i}=\Gamma_{i}dt$ 
and $\alpha=1/2$ (Stratonovich stochastic integral). \\
$\bullet$ Define $R^{j}_{k}=\partial_{k}y^{j}$, $j,k=1,...,n$, where $\boldsymbol{y}=\left(y^{1},...,y^{n}\right)$ is a coordinate system such that the map $\mathcal{T}:\boldsymbol{x}\rightarrow\boldsymbol{y}$ is a diffeomorphism.\\
$\bullet$ Let the equations of motion be 
\begin{equation}
X^{i}=\left(\mathcal{J}^{ij}-\gamma^{ij}\right)H_{0j}+\mathcal{J}^{ik}R^{j}_{k}\Gamma_{j},
\end{equation}
where the function $H\left(\boldsymbol{x},t\right)=H_{0}\left(\boldsymbol{x}\right)+y^{i}\Gamma_{i}\left(t\right)$ is the Hamiltonian of the system, $H_{0}\in C^{2}\lr{\mc{M}}$, 
and $\gamma^{ij}=\frac{1}{2}\beta \mathcal{J}^{ir}R^{k}_{r}\mathcal{J}^{js}R^{k}_{s}$ is the friction coefficient with $\beta\in\mathbb{R}$ a spatial constant.\\
$\bullet$ The corresponding transport equation for the probability distribution $f>0$ on a smoothly bounded domain $\Omega\subset\mathcal{M}$ with volume element $vol^{n}=dx^{1}\wedge ... \wedge dx^{n}$ is given by equation (\ref{FPE_6_2}). 
Suppose that on the boundary $\partial\Omega$ the conditions $Z\cdot N=0$ and $X_{0}\cdot N=0$ hold, with $Z$ the Fokker-Planck velocity such that $\partial_{t}f=-\partial_{i}\left(fZ^{i}\right)$, $X_{0}=\mathcal{J}^{ij}H_{0j}\partial_{i}$, and $N$ the outward normal to $\partial\Omega$.}

\textit{Then, the solution to (\ref{FPE_6_2}) is such that: 
\begin{equation}
\lim_{t\rightarrow\infty}\mathcal{J}\left(d\log f+\beta dH_{0}\right)=0~in~\Omega,\label{EqIM_2}
\end{equation}
for any choice of the coordinates $y^{j}$, $j=1,...,n$.}\\

Let us prove this statement. Recalling the expression of the Fokker-Planck velocity $Z$, equation \eqref{Z}, and 
setting $\alpha=1/2$ we obtain:
\begin{equation}
Z^{i}=\left(\mathcal{J}^{ij}-\gamma^{ij}\right)H_{0j}
-\frac{1}{2}\mathcal{J}^{ir}R^{k}_{r}\mathcal{J}^{js}R^{k}_{s}\frac{\partial\log f}{\partial x^{j}}.\label{Z2}
\end{equation}
In going from \eqref{Z} to this expression, we used the fact that $\mJ$ is measure preserving ($\partial_{i}\mJ^{ij}=0$, $j=1,...,n$) and that the matrix $R^{k}_{sj}=\p^{2} y^{k}/\p x^{s}\p x^{j}$ is symmetric so that $\mJ^{sj}R^{k}_{sj}=0$, $k=1,...,n$.
Consider now the following entropy functional:
\begin{equation}
S=-\int_{\Omega}f\log f\, vol^{n}.\label{S_6}
\end{equation}
The rate of change of $S$ is:
\begin{dmath}
\frac{dS}{dt}=
-\int_{\Omega}{\frac{\partial f}{\partial t}\left(1+\log f\right) \, vol^{n}}=
\int_{\Omega}{ f\frac{\partial Z^{i}}{\partial x^{i}} \, vol^{n}}
+\int_{\partial\Omega}{f\log f\, Z^{i}N_{i}\, dS^{n-1}}=
-\int_{\Omega}{f_{i}Z^{i} \, vol^{n}}.\label{dSdt_6}
\end{dmath}
Here we used the fact that $Z^{i}N_{i}$ vanish on the boundary $\partial\Omega$.
In this notation $N=N_{i}\partial_{i}$ is the outward normal to the bounding surface $\partial\Omega$ with surface element $dS^{n-1}$. 
Substituting (\ref{Z2}) in (\ref{dSdt_6}) we get:
\begin{dmath}
\frac{dS}{dt}
=\frac{1}{2}\int_{\Omega}{f_{i}\mathcal{J}^{ir}R^{k}_{r}\mathcal{J}^{js}R^{k}_{s}\left(\frac{\partial \log f}{\partial x^{j}}+\beta H_{0j}\right) \, vol^{n}}.\label{dSdt_6_2}
\end{dmath} 
Here we used the fact that $\mathcal{J}$ is measure preserving and thus the term involving $f_{i}\mathcal{J}^{ij}H_{0j}=\frac{\partial}{\partial x^{i}}\left(fX_{0}^{i}\right)$ can be written as a vanishing surface integral. 
Consider now conservation of total energy $E=\int_{\Omega}{f H_{0} \, vol^{n}}$:
\begin{dmath}
\frac{dE}{dt}=
\int_{\Omega}{  f\mathcal{J}^{ij}H_{0j}H_{0i}\, vol^{n}}
-\frac{1}{2}\int_{\Omega}{f\mathcal{J}^{ir}R^{k}_{r}\mathcal{J}^{js}R^{k}_{s}\left(\frac{\partial \log f}{\partial x^{j}}+\beta H_{0j}\right) H_{0i} \, vol^{n}}=
0.\label{dEdt_6}
\end{dmath}
Again, we used the fact that surface integrals vanish and the antisymmetry of $\mJ$.
This implies:
\begin{equation}
\int_{\Omega}{ \mathcal{J}^{ir}R^{k}_{r}\mathcal{J}^{js}R^{k}_{s}f_{j} H_{0i}\, vol^{n}}=-
\beta\int_{\Omega}{ f\left(\mathcal{J}^{ir}R^{k}_{r}H_{0i}\right)^{2}\, vol^{n}}.\label{beta}
\end{equation}
Observe that \eqref{beta} defines the spatial constant $\beta$ at each time $t$. 
Substituting this result in (\ref{dSdt_6_2}) 
and after some manipulations we obtain:
\begin{equation}
\frac{dS}{dt}=\frac{1}{2}\int_{\Omega}{ f\left[\mathcal{J}^{ir}R^{k}_{r}\lr{\frac{\partial\log f}{\partial x^{i}}+\beta H_{0i}}\right]^{2} \, vol^{n}}.\label{dSdt_6_4}
\end{equation}
In the limit of thermodynamic equilibrium we must have $\lim_{t\rightarrow\infty}dS/dt=0$.
Thus, for all non-zero $f$ 
one arrives at the result \eqref{EqIM_2}.
Notice that the matrix $R_{r}^{k}$ could be removed because the transformation $\mc{T}:\bol{x}\rightarrow\bol{y}$
is a diffeomorphism and is therefore invertible.

Let us make some considerations on the meaning and the physical implications of this result. 
The reason why equation \eqref{EqIM_2} holds is that $\mJ$ is measure preserving and $f$ is the probability distribution \textit{on} the invariant measure. 
Only in such coordinate system Shannon's entropy \eqref{S_6} has proper physical meaning, i.e. the entropy production represented by equation \eqref{dSdt_6_4} has a definite sign and therefore an extremum principle (maximum entropy) applies. If $g$ is the Jacobian of the coordinate change sending the invariant measure $vol^{n}$ to a different reference system $vol^{n}_{c}=g^{-1}vol^{n}$, the probability distribution in the new frame is $u=fg$. Here, the letter $c$ stands for Cartesian, since usually one is interested in the probability distribution observed in the Cartesian coordinate system of the laboratory frame. Define Shannon's entropy for the new distribution $u$ as $S_{c}=-\int_{\Omega}{u\log{u}}\,vol_{c}$. Then,  
the thermodynamically consistent entropy $\Sigma$ and the information measure $S_{c}$ are related as:
\begin{equation}
\Sigma
=S_{c}+\left\langle\log{g}\right\rangle,
\end{equation}
where the angle bracket stands for ensemble average.

It is useful to add some considerations on the boundary conditions $Z\cdot N=0$
and $X_{0}\cdot N=0$ on $\p\Omega$. Physically, they express the fact that probability does not escape
from the domain $\Omega$, and therefore the system can be considered as thermodynamically closed.
The condition $X_{0}\cdot N=0$ can be thought as a definition of the boundary itself, and can be satisfied, for example, by taking an Hamiltonian $H_{0}$ that is constant on the boundary, $H_{0i}=0$ on $\p\Omega$.
The condition $Z\cdot N=0$ is rather a boundary condition for $f$. If $H_{0i}=0$ on $\p\Omega$ 
 one can use the Neumann boundary condition $df=0$ on $\p\Omega$.

If the matrix $\mJ$ is invertible, equation \eqref{EqIM_2} becomes:
\begin{equation}
f^{eq}=\lim_{t\rightarrow\infty}f=A\exp\left\{-\beta H_{0}\right\}~in~\Omega,
\end{equation}
where $A\in\mathbb{R}_{>0}$ is a normalization constant.
Thereby, we can rephrase the result \eqref{EqIM_2} in the following way:
if the metric of space if current free, i.e. $\mc{O}^{n-1}=0$, and space is completely
accessible, i.e. $ker\lr{\mJ}=0$, the standard result of statistical mechanics apply on the 
invariant measure. The effect of a non-trivial kernel $ker\lr{\mJ}\neq 0$ can be understood 
with the next corollary of theorem \eqref{EqIM_2}.\\ 

\textit{Assume the hypothesis used to derive \eqref{EqIM_2}. 
In addition, assume that $\mJ$ has constant rank $2m=n-r$ and that it is a Poisson operator. 
Then, for every point $\bol{x}\in\Omega$ there exists a neighborhood $U\subset\Omega$ of $\bol{x}$ such that:
\begin{equation}
f^{eq}=\lim_{t\rightarrow\infty}f=A\exp\left\{-\beta H_{0}-\gamma\mc{F}\lr{\bol{C}}\right\}~in~U,\label{EqP_2}
\end{equation}
where $\gamma\in\mathbb{R}$ is a constant and $\mc{F}\lr{\bol{C}}$ an arbitrary function of the
$r$ Casimir invariants $\bol{C}=\lr{C^{1},...,C^{r}}$ whose gradients 
span the kernel of $\mJ$, i.e. $\mJ\lr{dC^{i}}=0$.}\\  
 
This result is a consequence of Darboux's theorem, according to which $\forall\bol{x}\in\Omega$ there exists a neighborhood $U\subset\Omega$ of $\bol{x}$ where we can find coordinates $\lr{u^{1},...,u^{2m},C^{1},...,C^{r}}$
such that the $C^{i}$ are Casimir invariants. 
Thus, the local solution to equation \eqref{EqIM_2} is of the form \eqref{EqP_2}.

In the case of a non-canonical Hamiltonian system, we see that statistical
equilibrium, which is achieved on the invariant measure assigned by Liouville's theorem, is determined by the energy $H_{0}$ and the Casimir invariants $C^{i}$.
In this way, the functions $C^{i}$ impart a non-trivial structure to the probability distribution $f$.
This type of self-organization is caused
by the existence of inaccessible regions in the phase space, which are mathematically represented by the fact that motion is restricted on the level sets of the Casimir invariants.

The last remark concerns the white noise assumption. 
This assumption must be justified on a case by case basis by showing that
the perturbations affecting a certain ensemble statistically behave as Gaussian white noise
in some appropriate coordinate system $\bol{y}$ (in the sense that the gradient $\p H_{I}/\p y^{r}$ 
of the interaction Hamiltonian $H_{I}$ with respect to the coordinates $\bol{y}$ 
can be considered as Gaussian white noise). 
In practice, using the invariant measure provided by the measure preserving operator,
one invokes the ergodic hypothesis by which ensemble and time averages can be interchanged. Then, 
fluctuations with vanishing ensemble averages can be conveniently represented as white noise processes of zero time average. Finally, notice that equation \eqref{EqIM_2} does not depend on the specific coordinates
$\bol{y}$. This means that, regardless of the coordinate frame where
a system is perturbed, statistical equilibrium is achieved on the invariant measure determined by $\mJ$.

\section{Diffusion with Beltrami operators}

We now move to operators that are not endowed with an invariant measure.
Specifically, we generalize equation (\ref{dSdt_2_1}) to $n$D. In this case we are interested in pure diffusion, i.e. $H_{0}=0$.
Then, from equation \eqref{EoMN2}, the relevant equation of motion reads:
\begin{equation}
X=\lr{\mJ^{ij}R_{j}^{r}\Gamma_{r}}\p_{i}.
\end{equation}
To further simplify the problem, set $R_{j}^{r}=\delta_{j}^{r}$. 
Recalling the transport equation \eqref{FPE_6} and putting $\alpha=1/2$, 
we arrive at the corresponding diffusion equation:
\begin{equation}
\frac{\partial f}{\partial t}=\frac{1}{2}\frac{\partial}{\partial x^{i}}\left[\mathcal{J}^{ik}\frac{\partial\left(\mathcal{J}^{jk}f\right)}{\partial x^{j}}\right]=\frac{1}{2}\lr{\Delta_{\perp}f+b^{i}f_{i}+\frac{1}{4}f\mf{B}}.\label{FPEDiff}
\end{equation}
Here, $\Delta_{\perp}f=\p_{i}\lr{\mJ^{ik}\mJ^{jk}f_{j}}$ is the $n$-dimensional normal Laplacian and $b^{i}=\mJ^{ik}\frac{\p\mJ^{jk}}{\p x^{j}}$. We have the following:\\

\textit{Assume that $\mJ\in\bigwedge^{2}T\mc{M}$ is a Beltrami operator ($\mf{B}=0$) on $vol^{n}=dx^{1}\w ...\w dx^{n}$. 
Consider the diffusion equation \eqref{FPEDiff} for the probability distribution
$f>0$ on a smoothly bounded domain $\Omega\subset\mc{M}$.
Assume the boundary conditions $Z\cdot N=0$ and $\bol{b}\cdot N=0$ on $\p\Omega$, where
$Z$ is the Fokker-Planck velocity such that $\p_{t}f=-\partial_{i}\lr{fZ^{i}}$, $\bol{b}=\mJ^{ik}\mJ^{jk}_{j}\p_{i}$, and $N$ the outward normal to $\p\Omega$.
Then,
\begin{equation}
\lim_{t\rightarrow\infty}\mJ\lr{d\log{f}}=0~in~\Omega.\label{DiffB_2}
\end{equation}}   

The proof can be given as follows. Consider the entropy functional:
\begin{equation}
S=-\int_{\Omega}{f\log{f}}\,vol^{n}.\label{SB}
\end{equation}
Following the same calculation of equation (\ref{dSdt_6}), the rate of change in $S$ is:
\begin{dmath}
\frac{dS}{dt}=
\frac{1}{2}\int_{\Omega}{\left[-\frac{f}{4}\mf{B}+f\left\lvert\mJ\lr{d\log{f}}\right\rvert^{2}\right]}\,vol^{n}=
\frac{1}{2}\int_{\Omega}{f\left\lvert\mJ\lr{d\log{f}}\right\rvert^{2}}\,vol^{n}.
\end{dmath}
Here, we used the boundary conditions to eliminate surface integrals and
the vanishing of $\mf{B}$. Then, since by hypothesis $f> 0$, one arrives at \eqref{DiffB_2}. 


As for theorem \eqref{EqIM_2}, the physical meaning of the requirements $Z\cdot N=0$ and $\bol{b}\cdot N=0$ on $\p\Omega$
is that probability does not escape from the boundaries.
If the diffusion equation is written in terms of the Cartesian coordinate system of $\mathbb{R}^{n}$,
the vector $\bol{b}$ corresponds to the field force $n-1$ form \eqref{bn-1} 
and, in $\mathbb{R}^3$, one obtains $\bol{b}=\bol{w}\cp\lr{\nabla\cp\bol{w}}$.
$\bol{b}$ acts as an effective drift. Indeed, from equation \eqref{FPEDiff}, one sees that
the Fokker-Plack velocity $Z$ can be decomposed as 
$2Z^{i}={f}^{-1}\mJ^{ik}\frac{\p\lr{\mJ^{jk}f}}{\p x^{j}}=
b^{i}+\mJ^{ik}\mJ^{jk}\frac{\p\log{f}}{\p x^{j}}$. 
Thus, $\bol{b}\cdot N=0$ on $\p\Omega$ means that the boundary must be chosen so that
the drift $\bol{b}$ does not transport any probability out of the domain $\Omega$.
The second condition $Z\cdot N=0$ can be intended as a boundary condition for the probability distribution $f$.
A possible way to satisfy these conditions is, for example, to assume that $\mJ$
is a strong Beltrami operator in a Cartesian coordinate system so that $\bol{b}=\bol{0}$,
and set $\nabla f=\bol{0}$ on $\p\Omega$. 

Equation \eqref{DiffB_2} 
says that the flat distrinution $f=$constant can be obtained even if no invariant measure exists.
In other words, the Beltrami operator is the largest class of antisymmetric operators such that the diffusion equation \eqref{FPEDiff} admits the solution $f=$constant. 
As already noted in section III, this fact can be verified by substituting the solution $f=$constant in equation \eqref{FPEDiff}. One obtains the condition $\mf{B}=0$. 
Beyond diffusion driven by Beltrami operators, the non-vanishing of $\mf{B}$ obstructs, in general, the determination of a suitable metric $g$ where an H-theorem can be obtained.
A possible way out is the extension method of equation \eqref{ExtJ}, which enables
the handling of a general antisymmetric operator by extending it to a measure preserving form. 
However, there are cases that can be solved explicitly even for $\mf{B}\neq 0$, as shown at the end of section III. 

\section{Conclusion}

In the present study we have investigated the properties of diffusion in systems that lack canonical phase space.
Such defect is caused by topological constraints that break the Hamiltonian structure of the dynamics and is mathematically represented by the violation of the Jacobi identity.
Under these circumstances, the usual arguments of statistical mechanics relying on the
invariant measure provided by Liouville's theorem do not apply, and diffusion causes, in general, the creation
of heterogeneous distributions.   

The characterization of diffusion processes in non-Hamiltonian ensembles
requires the determination of the regimes under which the law of maximum entropy holds.
While in Hamiltonian systems the sources of heterogeneity are either the special form of the energy
or the foliation of the phase space dictated by Casimir invariants, we have shown that in non-Hamiltonian
systems the determinant is the field charge, which measures the degree at which an antisymmetric operator (field tensor) departs from a Beltrami field.
We proved an H-theorem for systems characterized by a vanishing field charge, and
demonstrated the role of a finite field charge in generating heterogeneous structures.

In the generalization of the theory to arbitrary dimensions
we developed a geometrical classification of antisymmetric operators. 
Each of the new operators (measure preserving and Beltrami) introduced in this
study exhibits peculiar dynamical and statistical properties. 
We found that all antisymmetric
operators can be extended to a measure preserving form,
and that the standard results of statistical mechanics
can be generalized to the class of measure preserving
operators. This latter fact is remarkable, because such
operators do not posses an Hamiltonian structure. 

Finally, the normal
Laplacian is a novel object of mathematical
interest: this operator shows a clear interplay between
integrability in the context of differential geometry and
the study of non-elliptic PDEs.

\section{Acknowledgment}
The work of N.S. was supported by JSPS KAKENHI Grant No. 16J01486, and that of Z.Y. was supported by JSPS KAKENHI Grant No. 17H01177.
\end{normalsize}

\end{document}